\documentclass{article} % For LaTeX2e
\PassOptionsToPackage{usenames,dvipsnames,dvinames}{xcolor} 
\usepackage[usenames,dvipsnames,dvinames]{xcolor}
\usepackage{multirow}
\usepackage{microtype}
\usepackage{enumitem}
\usepackage{amsmath,amsthm}
\usepackage[mathscr]{eucal}
\usepackage{amssymb}
\usepackage{wrapfig}
\usepackage{footnote}
\usepackage{graphicx}
\usepackage[numbers]{natbib}
\usepackage{titlesec}
\newcommand{\truncated}{curve-normalized}
\newcommand{\Truncated}{Curve-normalized}
\newcommand{\Gs}{\mathcal{G}}
\newcommand{\Vs}{\mathcal{V}}
\newcommand{\Es}{\mathcal{E}}

\author{%
Rishabh Iyer$^{\dagger}$, Stefanie Jegelka$^{\ast}$, Jeff Bilmes$^{\dagger}$ \\
$^{\dagger}$ University of Washington, Dept.\ of EE, Seattle, U.S.A.\\$^{\ast}$ University of California, Dept. of EECS, Berkeley, U.S.A. \\
\texttt{rkiyer@uw.edu, stefje@eecs.berkeley.edu, bilmes@uw.edu}
}

\titlespacing{\paragraph}{%
  0pt}{%              left margin
  0.3\baselineskip}{% space before (vertical)
  1em}%   

\newcommand{\curv}{\ensuremath{\kappa_f}}
\newcommand{\curvf}[1]{\ensuremath{\kappa_{#1}}}

\DeclareMathOperator*{\argmin}{argmin}

\newtheorem{problem}{Problem}
\newtheorem{theorem}{Theorem}[section]
\newtheorem{observation}{Observation}[section]
\newtheorem{lemma}{Lemma}[section]
\newtheorem{proposition}[theorem]{Proposition}
\newtheorem{corollary}[theorem]{Corollary}

\usepackage{ifthen}
\newboolean{isdraft}
\setboolean{isdraft}{false} 
\ifthenelse{\boolean{isdraft}}{%
\newcommand{\myaddcomment}[3]{{\color{#1}{\ensuremath{\langle\!\!\langle}{\bf {#2} :} {#3}\ensuremath{\rangle\!\!\rangle}}}}
\newcommand{\rishabh}[1]{\myaddcomment{orange}{Rishabh}{#1}}
\newcommand{\JTR}[1]{\myaddcomment{orange}{Jeff\ensuremath{\rightarrow}Rishabh}{#1}}
\newcommand{\STR}[1]{\myaddcomment{orange}{Steffi\ensuremath{\rightarrow}Rishabh}{#1}}
\newcommand{\jeff}[1]{\myaddcomment{blue}{Jeff}{#1}}
\newcommand{\RTJ}[1]{\myaddcomment{blue}{Rishabh\ensuremath{\rightarrow}Jeff}{#1}}
\newcommand{\STJ}[1]{\myaddcomment{blue}{Steffi\ensuremath{\rightarrow}Jeff}{#1}}
\newcommand{\steffi}[1]{\myaddcomment{red}{Steffi}{#1}}
\newcommand{\JTS}[1]{\myaddcomment{red}{Jeff\ensuremath{\rightarrow}Steffi}{#1}}
\newcommand{\RTS}[1]{\myaddcomment{red}{Rishabh\ensuremath{\rightarrow}Steffi}{#1}}

}{
\newcommand{\rishabh}[1]{}
\newcommand{\JTR}[1]{}
\newcommand{\JTS}[1]{}
\newcommand{\jeff}[1]{}
\newcommand{\RTJ}[1]{}
\newcommand{\RTS}[1]{}
\newcommand{\toboth}[1]{}
\newcommand{\steffi}[1]{}
\newcommand{\STR}[1]{}
\newcommand{\STJ}[1]{}
}
\newcounter{propcounter}
\providecommand{\doarxiv}{true}
\usepackage{ifthen}
\newboolean{isarxiv}
\setboolean{isarxiv}{\doarxiv} 
\ifthenelse{\boolean{isarxiv}}{%
\newcommand{\arxiv}[1]{#1}
\newcommand{\notarxiv}[1]{}
}{
\newcommand{\arxiv}[1]{}
\newcommand{\notarxiv}[1]{#1}
}
\newcommand{\arxivalt}[2]{\ifthenelse{\boolean{isarxiv}}{#1}{#2}}

\notarxiv{}

\begin{document}

\title{Curvature and Optimal Algorithms for Learning and Minimizing Submodular Functions}

\maketitle
%%%%%%%%%%%%%%%%%%%%%%%%%%%%%%%%%%%%%%%%%%%%%%%%%%%%%%%%
\begin{abstract}

  We investigate three related and important problems
  connected to machine learning: approximating a submodular
  function everywhere, learning a submodular function (in a PAC-like
  setting~\cite{valiant1984theory}), and constrained minimization of submodular functions. 
  We show that the complexity of all three problems depends on the ``curvature'' of the submodular function, and provide lower and upper bounds that refine and improve previous results 
  %For all three problems, we provide improved bounds which depend on the
  %``curvature'' of a submodular function and improve on the previously
  %known best results for these
  %problems~
  \cite{balcanlearning,goel2009approximability, goemans2009approximating, 
    svitkina2008submodular}.
  Our proof techniques are fairly generic. We either use a black-box transformation of the function (for approximation and learning), or a transformation of algorithms to use an appropriate surrogate function (for minimization).
  %when the function
  %is not too curved -- a property that holds for many real-world
  %submodular functions. In the former two problems, we obtain these
  %bounds through a generic black-box transformation (which can
  %potentially work for any algorithm), while in the case of submodular
  %minimization, we propose a framework of algorithms which depend on
  %choosing an appropriate surrogate for the submodular function. In
  %all these cases, we provide almost matching lower bounds.
%While in the worst case, this bound matches the previously known results of~\cite{goemans2009approximating} and~\cite{balcanlearning}, it improves the factor for the class of submodular functions which are not too curved. In both cases, we provide almost matching lower bounds, thus showing our analysis is tight.\looseness-1
%We then also provide a framework for constrained submodular function minimization problems, which depends on choosing an appropriate approximation of the function, which can be easily minimized over the constraints. We show that with particular choices of the approximations, we can obtain tight approximation factors for a large class of these problems.
%Our bounds are curvature dependent, thereby providing improved results for the class of submodular functions with restricted curvature. Moreover, in every case, we show matching lower bounds. \looseness-1
  Curiously, curvature has been known to influence approximations for
  submodular maximization \cite{conforti1984submodular,
    vondrak2010submodularity}, but its effect on minimization,
  approximation and learning has hitherto been open. We complete this
% dont' use a double metaphor "close the gap to complete a picture :-)"
  picture, and also support our theoretical claims by empirical
  results.
\end{abstract}\looseness-1

%%%%%%%%%%%%%%%%%%%%%%%%%%%%%%%%%%%%%%%%%%%%%%%%%%%%%%%%%%%%%%%%%%%
\section{Introduction}
\label{sec:introduction}

Submodularity is a pervasive and important property
in the areas of combinatorial optimization, economics, operations
research, and game theory. In recent years, submodularity's use
in machine learning has begun to proliferate as well.  A set function
$f: 2^V \to \mathbb R$ over a finite set $V = \{1, 2, \ldots, n\}$ is
\emph{submodular} if for all subsets $S, T \subseteq V$, it holds that
$f(S) + f(T) \geq f(S \cup T) + f(S \cap T)$. Given a set $S \subseteq
V$, we define the \emph{gain} of an element $j \notin S$ in the
context $S$ as $f(j | S) \triangleq f(S \cup j) - f(S)$. A function $f$ is
submodular if it satisfies \emph{diminishing marginal returns},
namely $f(j | S) \geq f(j | T)$ for all $S \subseteq T, j \notin T$,
and is \emph{monotone} if $f(j | S) \geq 0$ for all $j \notin S, S
\subseteq V$. %\looseness-1

\arxiv{The search for optimal algorithms for submodular optimization
  has seen substantial progress~\cite{fujishige2005submodular,
    iwata2008submodular, feldman2012optimal} in recent years, but is still an ongoing
  endeavor.  The first polynomial-time algorithm used the ellipsoid
  method \cite{grotschel1981ellipsoid,grotschel1984geometric}, and
  several combinatorial algorithms followed
  \cite{iwata2001combinatorial, iwata2002fully, fleischer2003push,
    iwata2003faster, iwata2009simple, orlin2009faster}.
For a detailed summary,
see~\cite{iwata2008submodular}. % While the above algorithms refer to the ``easy'' case of minimizing a submodular function without constraints, most constrained submodular minimization problems are extremely hard and have polynomial lower bounds.
% Similarly, the
% problems of approximating and learning monotone submodular functionsAlgorithmic
% 
% everywhere have been shown to be very hard.
Unlike submodular minimization, submodular maximization is NP
hard. However, maximization problems admit constant-factor
approximations~\cite{nemhauser1978, sviridenko2004note, janvondrak,
  feldman2012optimal}, often even in the constrained case
\cite{sviridenko2004note, lee2009non, matroidimproved,
  chekuri2011submodular, feldman2011unified, vondrakcontinuousgreedy}.
% For example, a $1 - 1/e$
% approximation algorithm was shown for monotone submodular maximization
% subject to a cardinality constraint~\cite{nemhauser1978}. This result is tight under
% the value oracle model~\cite{feige1998threshold,nemhauser78}.
% Moreover, a recent paper by~\cite{feldman2012optimal} has shown a
% $1/2$ linear time approximation algorithm for unconstrained submodular
% maximization, also tight under the value oracle
% model~\cite{janvondrak}.  A number of analogous approximation
% guarantees have been shown for various forms of constrained submodular
% maximization like knapsack and matroid independence
% constraints~\cite{sviridenko2004note, lee2009non, matroidimproved,
%   chekuri2011submodular, feldman2011unified, vondrakcontinuousgreedy}.
}

While submodularity, like convexity, occurs naturally in a wide variety of
problems, recent studies have shown that in the general case, many
submodular problems of interest are very hard: the problems of
learning a submodular function or of submodular minimization
under constraints do not even admit constant or logarithmic
approximation factors in polynomial time
\cite{balcanlearning,goel2009optimal,goemans2009approximating,iwata2009submodular,svitkina2008submodular}.
These rather pessimistic results however stand in sharp contrast to empirical observations, which suggest
% Nevertheless, empirical results suggest 
that these lower bounds are
specific to rather contrived classes of functions, whereas
much better results can be achieved in many practically relevant cases.
%functions.  
Given the increasing importance of submodular
functions in machine learning, these observations beg the question of
qualifying and quantifying properties that make sub-classes of
submodular functions more amenable to learning and
optimization. Indeed, limited prior work has shown improved
results for constrained minimization and learning of
 sub-classes of submodular functions, including symmetric
functions \cite{balcanlearning,sotoG10}, concave functions
\cite{goel2009optimal,kohliOJ13,nikolova2010approximation},
\emph{label cost} or covering functions
\cite{hassinMS07,zhang2011approximation}. %\looseness-1
%or covering functions (\steffi{TODO: learning coverage functions}).

In this paper, we take additional steps towards addressing 
the above problems and show how the generic notion of the
\emph{curvature} -- the deviation from modularity-- of a submodular function determines both upper and
lower bounds on approximation factors for many learning and
constrained optimization problems. In particular, our quantification tightens
the generic, function-independent bounds in \cite{goemans2009approximating,balcanlearning,svitkina2008submodular,goel2009optimal,iwata2009submodular} for many practically relevant functions.
Previously, the concept of
curvature has been used to tighten bounds for submodular maximization
problems~\cite{conforti1984submodular, vondrak2010submodularity}. Hence, our results complete a unifying picture of the effect of curvature on submodular problems. By
 quantifying the influence of curvature on other problems, we improve previous bounds in
 \cite{goemans2009approximating,balcanlearning,svitkina2008submodular,goel2009optimal,iwata2009submodular}
 for many functions used in applications.  Curvature, moreover, does
 not rely on a specific functional form but generically
 only on the marginal gains. It allows a smooth transition between the `easy' functions and the `really hard' subclasses of submodular functions.\looseness-1
 
% \looseness-1

%%%%%%%%%%%%%%%%%%%%%%%%%%%%%%%%%%%%%%%%%%%%%%%%%%%%%%%%%%%%%
\section{Problem statements, definitions and background}
\label{sec:probl-stat-defin}

Before stating our main results, we provide some necessary definitions and
% first define detail our problems, curvature and 
introduce a new
concept, the \emph{curve normalized} version of a submodular
function.
  Throughout this paper, we assume that the submodular function $f$ is
  defined on a ground set $V$ of $n$ elements, that it is
  nonnegative and $f(\emptyset) = 0$. We also use normalized
  \emph{modular} (or additive) functions $w: 2^V \to \mathbb{R}$ which
  are those that can be written as a sum of weights, $w(S) = \sum_{i
    \in S}w(i)$.
% and describing a few additional requisite concepts before
%presenting our main results starting in \S\ref{learnimproved}. 
We are
concerned with the following three problems: % (where $n = |V|$): %\looseness-1

\begin{problem}{(Approximation \cite{goemans2009approximating})}\label{prob1}
  Given a submodular function $f$ in form of a value oracle, find an
  approximation $\hat{f}$ (within polynomial time and representable
  within polynomial space), such that for all $X \subseteq V$, it
  holds that $\hat{f}(X) \leq f(X) \leq \alpha_1(n) \hat{f}(X)$ for a
  polynomial $\alpha_1(n)$.
\end{problem}
\begin{problem}{(PMAC-Learning \cite{balcanlearning})}\label{prob2}
  Given i.i.d\ training samples $\{(X_i,f(X_i)\}_{i=1}^m$ from a distribution $\mathcal D$, learn an approximation $\hat{f}(X)$ that is, with probability $1-\delta$, within a multiplicative factor of $\alpha_2(n)$ from $f$.
  \arxiv{PMAC learning is defined like PAC learning with the added relaxation that the function is, with high probability, approximated within a factor of $\alpha_2(n)$.}
  % \steffi{put formal PMAC-definition here?}\JTS{Yes, in extended version.}
  \end{problem}
  \begin{problem}{(Constrained optimization \cite{svitkina2008submodular,goel2009optimal,iwata2009submodular, jegelka2011-inference-gen-graph-cuts})}\label{prob:min}
    Minimize a submodular function $f$ over a family $\mathcal
    C$ of feasible sets, i.e., $\min_{X \in \mathcal C}f(X)$. 
    % $\mathcal C$ is the set
    % of feasible sets, also called combinatorial constraints.
  \end{problem}

  In its general form, the approximation problem was first studied by
  \citet{goemans2009approximating}, who approximate any
  monotone submodular function to within a factor of $O(\sqrt{n}\log
  n)$, with a lower bound of $\alpha_1(n) = \Omega(\sqrt{n}/\log
  n)$. Building on this result,
  %This approximation result has been a key ingredient for work on
  %learning submodular functions. An algorithm by
  \citet{balcanlearning} show how to PMAC-learn a monotone submodular function
  within a factor of $\alpha_2(n) = O(\sqrt{n})$, and prove % The same work shows
  a lower bound of $\Omega(n^{1/3})$ for the learning
  problem. Subsequent work extends these results to sub-additive
  and fractionally sub-additive functions~\arxivalt{\cite{balcan2011learning,badanidiyuru2012sketching}}{\cite{balcan2011learning}}. Better learning results are
  possible for the subclass of \emph{submodular shells}
  \cite{lin2012learning} and Fourier sparse set functions
  \cite{stobbe12learning}. \arxiv{Very recently Devanur \textit{et al}~\cite{devanur2013approximation} investigated a related problem of approximating one class of submodular functions with another and they show how many non-monotone submodular functions can be approximated with simple directed graph cuts within a factor of $n^2/4$ which is tight. They also consider problems of approximating symmetric submodular functions and other subclasses of submodular functions. 

}Both Problems 1 and 2 have numerous
  applications in algorithmic game theory and
  economics~\cite{balcanlearning, goemans2009approximating} as well as
  machine learning~\cite{balcanlearning,lin2011-class-submod-sum,
    lin2012learning,stobbe12learning,jegelka2010online}.  \arxiv{For example, applications like bundle pricing, predicting prices of objects or growth rates etc. often have diminishing returns and a natural problem is to estimate these functions~\cite{balcanlearning}. Similarly in machine learning, a number of problems involving sensor placement, summarization and others~\cite{lin2011-class-submod-sum, rkiyeruai2012} can be modeled through submodular functions. Often in these scenarios we would want to explicitly approximate or learn the true objective. For example in the case of document summarization, we are given the ROUGE scores. Since this function is submodular~\cite{lin2011-class-submod-sum}, a natural application is to learn these functions for summarization tasks.}\looseness-1
    
    % % or label cost functions \cite{jegelka2010online}

  Constrained submodular minimization arises in applications
  such as power assignment or transportation problems \arxivalt{
    \cite{krause2005near,krause06near,wan02networks,rkiyersemiframework2013}
  }{\cite{krause2005near, wan02networks, rkiyersemiframework2013}}.  In machine learning, it
  occurs, for instance, in the form of MAP inference in high-order
  graphical models
  \cite{\arxiv{delong2012minimizing,vicente2008graph,}jegelka2011-nonsubmod-vision} or in
  size-constrained corpus extraction \cite{lin11}.  Recent results
  show that almost all constraints make it hard to solve the
  minimization even within a constant factor
  \cite{svitkina2008submodular,goel2009approximability,jegelka2011-inference-gen-graph-cuts}. 
  Here, we will focus
  on the constraint of imposing a lower bound on the cardinality, 
  %of
  %the solution 
  and on combinatorial constraints where $\mathcal{C}$
  is the set of all $s$-$t$ paths, $s$-$t$ cuts, spanning trees, or
  perfect matchings in a graph.

\arxiv{\subsection{Curvature of a Submodular function}}
A central concept in this work is the total \emph{curvature} $\curv$ of a submodular function
$f$ and the curvature $\curv(S)$ with respect to a set $S \subseteq V$, defined as~\cite{conforti1984submodular, vondrak2010submodularity}
\begin{align}
  \label{eq:1}
% \text{{\bf Curvature} of $f$:}
  \curv = 1 - \min_{j \in V}\frac{f(j\mid V\setminus j)}{f(j)},\qquad\quad \curv(S) 
= 1 - \min_{j \in S} \frac{f(j | S \backslash j)}{f(j)}.
\end{align}
Without loss of generality, assume that $f(j) > 0$ for all $j \in V$.\arxiv{ This follows since, if there exists an element $j \in V$ such that $f(j) = 0$, we can safely remove element $j$ from the ground set, since for every set $X$, $f(j | X)$ = 0 (from submodularity), and including or excluding $j$ does not make any difference to the cost function.}
\arxivalt{

We also define two alternate notions of curvature. Define $\hat{\kappa_f}(S)$ and $\tilde{\kappa_f}(S)$ as,
\begin{align}
\hat{\kappa_f}(S) = 1 - \frac{\sum_{j \in S} f(j | S \backslash j)}{\sum_{j \in S} f(j)}, \tilde{\kappa_f}(S) = 1 - \min_{T \subseteq V} \frac{f(T | S) + \sum_{j \in S \cup T} f(j | S \cup T \backslash j)}{f(T)}
\end{align}
These different forms of curvature are closely related.
\begin{proposition}
For any monotone submodular function and set $S \subseteq V$,
\begin{align}
\hat{\kappa_f}(S) \leq \kappa_f(S) \leq \tilde{\kappa_f}(S) \leq \kappa_f
\end{align}
\end{proposition}
\begin{proof}
It is easy to see that $\curv(S) \leq \curv(V) = \curv$, by the fact that $\curv(S)$ is a monotone-decreasing set function. To show that $\kappa_f(S) \leq \tilde{\kappa_f}(S)$, note that,
\begin{align}
\tilde{\kappa_f}(S) &= \min_{T \subseteq V} \frac{f(T | S) + \sum_{j \in S \cup T} f(j | S \cup T \backslash j)}{f(T)} \nonumber \\
		&\geq \min_{T \subseteq V: |T| = 1} \frac{f(T | S) + \sum_{j \in S \cup T} f(j | S \cup T \backslash j)}{f(T)} \nonumber \\
		&\geq \min\{\min_{j \in S} \frac{f(j | S \backslash j)}{f(j)}, \min_{j \notin S} \frac{f(j | S)}{f(j)}\} \nonumber \\
		&\geq \min_{j \in S} \frac{f(j | S \backslash j)}{f(j)} \nonumber \\
		&\geq 1 - \kappa_f(S)
\end{align}
We finally prove that $\hat{\kappa_f}(S) \leq \kappa_f(S)$. 
%To show this consider two sequences $\{a_1, \cdots, a_k\}$ and $\{b_1, \cdots, b_k\}$. We show that 
%\begin{align}
%m = \min_i \frac{a_i}{b_i} \leq \frac{\sum_i a_i}{\sum_i b_i}.
%\end{align}
%Notice that, $\forall i, m b_i \leq a_i$. Hence $\sum_i a_i \geq m \sum_i b_i$, which shows that $\min_i \frac{a_i}{b_i} \leq \frac{\sum_i a_i}{\sum_i b_i}$. Direct application of this fact shows the result.
Note that,
\begin{align*}
1 - \kappa_f(S) &= \min_{j \in S} \frac{f(j | S \backslash j)}{f(j)}\\
	&\leq \frac{f(j | S \backslash j)}{f(j)}, \forall j \in S
\end{align*}
Also notice that,
\begin{align*}
1 - \hat{\kappa_f}(S) &= \frac{\sum_{j \in S} f(j | S \backslash j)}{\sum_{j \in S} f(j)}\\
		&\geq \frac{\sum_{j \in S} (1 - \kappa_f(S)) f(j)}{\sum_{j \in S} f(j)} \\
	&\geq 1 - \kappa_f(S)
\end{align*}
Hence, $\hat{\kappa_f}(S) \leq \kappa_f(S)$.
\end{proof}
Hence $\hat{\kappa_f}(S)$ is the tightest notion of curvature. In this paper, we shall see these different notions of curvature coming up in different bounds.
}{It is easy to see that $\curv(S) \leq \curv(V) = \curv$, and hence $\curv(S)$ is a tighter notion of curvature.} A modular function has curvature $\curv = 0$, and a matroid rank
function has maximal curvature $\curv = 1$. Intuitively, $\curv$
measures how far away $f$ is from being \emph{modular}. 
%\notarxiv{
Conceptually, 
curvature is distinct from the recently
  proposed \emph{submodularity ratio} \cite{das2011submodular} that
  measures how far a function is from being \emph{submodular}. 
  %--- curvature
  %measures how close a submodular function is to being
  %\emph{modular}.
%} %\looseness-1
%Analogous to
%\cite{conforti1984submodular, vondrak2010submodularity}, we also use a
%tighter variant of curvature that depends on a given set $S$:
Curvature has served to
tighten bounds for submodular maximization problems, e.g., from
$(1-1/e)$ to $\frac{1}{\curv}(1 - e^{-\curv})$ for monotone submodular
maximization subject to a cardinality constraint
\cite{conforti1984submodular} or matroid constraints
\cite{vondrak2010submodularity}, and these results are tight. \arxiv{In fact, \cite{vondrak2010submodularity} showed that this result for submodular maximization holds for the tighter version of curvature $\tilde{\kappa_f}(S^*)$, where $S^*$ is the optimal solution. In other words, the bound for the greedy algorithm of \cite{vondrak2010submodularity} can be tightened to $\frac{1}{\tilde{\curv}(S^*)}(1 - e^{-\tilde{\curv}(S^*)})$.

}For
submodular minimization, learning, and approximation,
however, the role of curvature has not yet been addressed (an
exception are the upper bounds in \cite{rkiyersemiframework2013} for
minimization). In the following sections, we complete the picture of how curvature affects the complexity of submodular maximization and minimization, approximation, and learning.

%address minimization, learning, and
%approximation, and provide curvature-dependent upper \emph{and} lower
%bounds for all three of our problems. \arxiv{We also show in certain cases, that the bounds involve many of the tighter notions of curvature, and using these we demonstrate theoretically how a number of real world submodular functions admit much tighter approximation guarantees.}

%% Wed May 29 17:00:42 2013: TODO, make sure
%% these points are in the uncommented paragraph below it.
% This paper studies the effect of curvature on the minimization,
% learning and approximation of submodular functions and thereby
% complements known results for submodular minimization.  Our results
% also provide a theoretical foundation for the observation that in
% practice, many submodular problems are much more benign than
% previously established lower bounds.  In many cases, they replace
% known polynomial bounds by curvature-dependent constant factors.

The above-cited lower bounds for Problems \ref{prob1}--\ref{prob:min} were established with functions of maximal curvature
($\curv=1$) which, as we will see, is the worst case. 
By contrast,
many practically interesting functions have smaller curvature, and 
%(as we show in this paper) 
our analysis will provide an explanation for the good empirical results observed with such functions
%admit better approximations, a fact that has been
%observed in practice too 
\cite{rkiyersemiframework2013,lin2011-class-submod-sum,jegelkathesis}. An
example for functions with $\curv < 1$ is the class of concave
over modular functions that have been used in speech\STR{Is this ``speeech''?} processing
\cite{lin2011-class-submod-sum} and computer
vision~\cite{jegelka2011-nonsubmod-vision}. This class comprises, for
instance, functions of the form $f(X) = \sum_{i = 1}^k (w_i(X))^a$, for some $a \in [0,
1]$ and a nonnegative weight vectors $w_i$. Such functions may be defined over clusters $C_i \subseteq V$, in which case the weights $w_i(j)$ are nonzero only if $j \in C_i$~\cite{lin2011-class-submod-sum,jegelka2011-nonsubmod-vision, rkiyeruai2012}.
% We shall see that this class of functions provably be better approximated and minimized than general submodular functions. 
%\STR{We could also
%  instead make a small table of example functions and their
%  curvature.}\RTS{I haven’t done that yet.. Lets keep that only if we are desperate for space.}
%

\arxiv{A related quantity distinct from curvature that  has been
  introduced in the machine learning community is the
  \emph{submodularity ratio}~\cite{das2011submodular}:
\begin{align}
\gamma_{U,k}(f)
= \min_{L \subseteq U, S : |S| \leq k, S \cap L = \emptyset}
\frac{ \sum_{x \in S} f(x|L) }{f(S|L)}
\end{align}
This parameter shows the decay of approximation bounds when an
algorithm for submodular maximization is applied to non-submodular
functions. The submodularity ratio measures how ``close'' $f$ is to
submodularity, and helps characterize theoretical bounds for functions
which are approximately submodular. Curvature, by contrast, 
measures how close a submodular function to being modular.}

%%%%%%%%%%%%%%%%%%%%%%%%%%%%%%%%%%%%%%%%%%%%%%%%%%%%%%%%%%%%%%%%%%%%%
\arxiv{\subsection{The \Truncated{} Polymatroid function}
\label{hardnessframework}}
\notarxiv{\paragraph{Curvature-dependent analysis.}}
To analyze Problems~\ref{prob1} -- \ref{prob:min},
we introduce the concept of a %we call the
\emph{\truncated{}}
polymatroid\footnote{A polymatroid function is a monotone increasing,
nonnegative, submodular function satisfying $f(\emptyset) = 0$.}.
%function in order to address
%Problems~\ref{prob1} -- \ref{prob:min}. 
Specifically, we define
the $\curv$-\emph{\truncated{}} version of $f$ as
%, which shall play crucial roles in obtaining improved guarantees for the problems we consider in this paper:
\begin{align}
\label{eq:defineg}
f^{\kappa}(X) = 
\frac{
f(X) - {(1 - \curv)} \sum_{j \in X} f(j)
}{\curv}
\end{align}
If $\curv=0$, then we set $f^\kappa \equiv 0$.  
We call $f^{\kappa}$ the \truncated{} version of $f$ because its curvature is
$\curvf{f^\kappa} = 1$.  The function $f^\kappa$ allows us to decompose a submodular
% Moreover, we may decompose a submodular
function $f$ into a ``difficult'' polymatroid function and an ``easy''
modular part as $f(X) = f_\text{difficult}(X) +
m_\text{easy}(X)$ where $f_\text{difficult}(X) = \kappa_f f^\kappa(X)$
and $m_\text{easy}(X) = (1-\kappa_f)\sum_{j \in X} f(j)$.  Moreover, we may modulate the curvature of
given any function $g$ with $\curvf{g} = 1$, by constructing a function %$f$ constructed as
$f(X) \triangleq c g(X) + (1- c) |X|$ with curvature $\curvf{f} = c$ but
otherwise the same polymatroidal structure as $g$. %\looseness-1

Our curvature-based decomposition is different from decompositions such as that into a \emph{totally normalized} function and a modular function \cite{cun82}. Indeed, the \truncated{} function has some specific properties that will be useful later on
% There are previous forms of submodular decomposition, for
% example the \emph{totally normalized} function of \cite{cun82} where
% an arbitrary submodular function is expressed as a (totally
% normalized) polymatroid function plus a modular function. In our current case,
% however, we will see that the \truncated{} function has some important
% and distinct properties.  For example, our approximations work on 
% the polymatroid part, and the added modularity adjusts the hardness and
% improves approximation bounds. Indeed, we immediately get the
% following results 
%for curature and \truncated{} submodular functions
\notarxiv{(proved in~\cite{nips2013extendedvcurv})}:
\begin{lemma}\label{lem:simplefacts}
If $f$ is monotone submodular with $\curv > 0$, then % it holds that % 
\arxivalt{%the following holds:\looseness-1
\begin{align}
  f(X) \leq \sum_{j \in X} f(j),\,\,\ f(X) %\geq (1 - \curv(X)) \sum_{j \in X} f(j) 
\geq (1 - \curv) \sum_{j \in X} f(j).
\end{align}}{$f(X) \leq \sum_{j \in X} f(j)$ and $f(X) \geq (1 - \curv) \sum_{j \in X} f(j)$.}
\end{lemma}
\arxiv{\begin{proof}
The inequalities follow from
%We use the 
submodularity and monotonicity of $f$. The first part follows from the subadditivity of $f$. The second inequality follows since $f(X) \geq \sum_{j \in X} f(j | V \backslash j) \geq (1 - \curv) \sum_{j \in X} f(j)$, since $\forall j \in X, f(j | V \backslash j) \geq (1 - \curv) f(j)$ by definition of $\curv$. %The last part of the inequality follows directly from the fact that $\curv(X) \leq \curv$.
\end{proof}}
\begin{lemma}
If $f$ is monotone submodular, then
% For monotone submodular function $f$ with curvature ${\curv}$, 
$f^{\kappa}(X)$ in Eqn.~\eqref{eq:defineg}
is a monotone non-negative submodular function. Furthermore, $f^{\kappa}(X) \leq \sum_{j \in X} f(j)$.
\end{lemma}
\arxiv{
\begin{proof}
 Submodularity of $f^{\kappa}$ is evident from the definition. To show the monotonicity, it suffices to show that $f(X) - {(1 - \curv)} \sum_{j \in X} f(j)$ is monotone non-decreasing
  and non-negative submodular. To show it is non-decreasing, notice that $\forall j \notin X, f(j | V \backslash j) -
  (1 - \curv) f(j)  \geq 0$, since $(1 - \curv) f(j) \leq f(j | V \backslash j)$ by the definition of $\curv$. 
  %Similarly $\forall j \in X, f(j | X \backslash j) - (1 - \curv(X)) f(j) \geq 0$ again by the same reason.
  % Hence $g(X)$ is non-decreasing. Notice that it is also non-negative
  % since $g(X) \geq g(\emptyset) \geq 0$. 
  Non-negativity follows from monotonicity and the fact that $f^{\kappa}(\emptyset) = 0$.
  %\JTR{I think we also need to assume normalized for this as in $f(\emptyset) \geq 0$? Should explicitely say somewhere above. }%
%  The remainder is implied by Lemma~\ref{lem:simplefacts}.
  To show the second part,
  notice that $\frac{f(X) - (1 - \curv) \sum_{j \in X} f(j)}{\curv} \leq \frac{\sum_{j \in X} f(j) -  (1 - \curv) \sum_{j \in X} f(j)}{\curv} = \sum_{j \in X} f(j)$.
\end{proof}}

%%%%%%%%%%%%%%%%%%%%%%%%%%%%%%%%%%%%%%%%%%%%%%%%%%%%%%%%%%%%%%%%%%%%%
\arxiv{\subsection{A framework for curvature-dependent lower bounds.}
\label{hardnessframework}}

The function $f^\kappa$ will be our tool for analyzing the hardness of submodular problems.
%We also provide a new framework to analyze the hardness of these problems. 
Previous information-theoretic lower bounds for
Problems~\ref{prob1}--\ref{prob:min}
\cite{goel2009approximability,goemans2009approximating,iwata2009submodular,svitkina2008submodular}
are \emph{independent} of curvature and use functions with
$\curv = 1$.  These curvature-independent bounds are proven by
constructing two essentially indistinguishable matroid rank functions
$h$ and $f^R$, one of which depends on a random set $R \subseteq
V$. One then argues that any algorithm would need to make a
super-polynomial number of queries to the functions for being able to
distinguish $h$ and $f^R$ with high enough probability. The lower
bound will be the ratio $\max_{X \in \mathcal{C}} h(X) / f^R(X)$.  We
extend this proof technique to functions with a fixed given curvature.
To this end, we define the functions
\begin{equation}
f_{\kappa}^R(X) = \curv f^R(X) + (1 - \curv) |X| \quad \text{ and }\quad h_{\kappa}(X) =  \curv h(X) + (1 - \curv) |X|.\label{eq:hidingfuncs}
\end{equation}
Both of these functions have curvature $\curv$.  This
construction enables us to explicitly introduce the effect of
curvature into information-theoretic bounds for all monotone
submodular functions.  

\paragraph{Main results.} The curve normalization \eqref{eq:defineg} leads to refined upper bounds for Problems~\ref{prob1}--\ref{prob:min}, while the curvature modulation~\eqref{eq:hidingfuncs} provides matching lower bounds. The following are some of our main results:
for approximating submodular functions (Problem~\ref{prob1}), we
replace the known bound 
of $\alpha_1(n) = O(\sqrt{n} \log n)$ \cite{goemans2009approximating} by
an improved curvature-dependent $O(\frac{\sqrt{n}
  \log n}{1 + (\sqrt{n} \log n - 1) (1 - \curv)})$. We complement this
with a lower bound of $\tilde{\Omega}(\frac{\sqrt{n}}{1 + (\sqrt{n} -
  1)(1 - \curv)})$. For learning submodular functions
(Problem~\ref{prob2}), we refine the known bound of $\alpha_2(n) =
O(\sqrt{n})$ \cite{balcanlearning} in the PMAC setting to a curvature
dependent bound of $O(\frac{\sqrt{n}}{1 + (\sqrt{n} - 1)
  (1 - \curv)})$, with a lower bound of
$\tilde{\Omega}(\frac{n^{1/3}}{1 + (n^{1/3} - 1) (1 -
  \curv)})$. Finally, Table~\ref{tab:results} summarizes our
curvature-dependent approximation bounds for constrained minimization
(Problem~\ref{prob:min}). These bounds refine many of the results in
\cite{goel2009approximability, svitkina2008submodular,
  iwata2009submodular,jegelka2011-inference-gen-graph-cuts}.  In
general, our new curvature-dependent upper and lower bounds refine
known theoretical results whenever $\curv < 1$, in many cases
replacing known polynomial bounds by a curvature-dependent constant
factor $1/(1 - \kappa_f)$. Besides making these bounds precise, the decomposition and the \truncated{} version
\eqref{eq:defineg} 
% not only makes the lower bounds precise, it also 
are the basis for constructing tight algorithms that (up to logarithmic
factors) achieve the lower bounds.

\begin{table}
  \centering
%  \STR{scriptsize is extremely painful to read -- the reader is our friend, not enemy :)}
  %\small
  \begin{tabular}{|l|c|c|c|c|}
  \hline
    Constraint & Modular approx. (MUB) & Ellipsoid approx. (EA) & Lower bound 
    % \footnote{$\tilde{\Omega}$ signifies that we neglect $\log$ factors.}   & Abs. Hardness 
    \\ \hline
   Card. LB & $\frac{k}{1+(k-1)(1 - \curv)}$ & $O(\frac{\sqrt{n} \log n}{1 + (\sqrt{n} \log n - 1)(1 - \curv)})$ 
   %& - 
   & $\tilde{\Omega}(\frac{n^{1/2})}{1 + (n^{1/2} - 1)(1 - \curv)})$ %& $\Omega(\sqrt{n/\log n})$  
   \\
%  Matroid Span & $\frac{k}{1+(k-1)(1 - \curv)}$ & $O(\frac{\sqrt{n} \log n}{1 + (\sqrt{n} \log n - 1)(1 - \curv)})$ & - & $\tilde{\Omega}(\frac{n^{1/2})}{1 + (n^{1/2} - 1)(1 - \curv)})$ & $\tilde{\Omega}(\sqrt{n})$ \\
   Spanning Tree & $\frac{n}{1+(n - 1)(1 - \curv)}$ & $O(\frac{\sqrt{m} \log m}{1 + (\sqrt{m}\log m - 1)(1 - \curv)})$ 
   %& - 
   & $\tilde{\Omega}(\frac{n}{1 + (n - 1)(1 - \curv)})$ %& $\tilde{\Omega}(n)$  
   \\
   Matchings & $\frac{n}{2+(n - 2)(1 - \curv)}$ & $O(\frac{\sqrt{m} \log m}{1 + (\sqrt{m} \log m - 1)(1 - \curv)})$ 
   %& - 
   & $\tilde{\Omega}(\frac{n}{1 + (n - 1)(1 - \curv)})$ %& $\tilde{\Omega}(n)$  
   \\  
\arxiv{   Edge Cover & $\frac{n}{1+(\frac{n}{2} - 1)(1 - \curv)}$ & $O(\frac{\sqrt{m} \log m}{1 + (\sqrt{m} \log m - 1)(1 - \curv)})$ 
  %& - 
  & $\tilde{\Omega}(\frac{n}{1 + (0.5 n - 1)(1 - \curv)}$ %& $\tilde{\Omega}(n)$  
\\  }
     s-t path & $\frac{n}{1 + (n-1)(1-  \curv)}$ & $O(\frac{\sqrt{m}\log m}{1 + (\sqrt{m}\log m - 1)(1 - \curv)})$ 
     % & $O(\frac{n^{2/3}\log n}{1 + (n^{2/3}\log n - 1)(1 - \curv)})$ 
     & $\tilde{\Omega}(\frac{n^{2/3}}{1 + (n^{2/3} - 1)(1 - \curv)})$ %& $\tilde{\Omega}(n^{2/3})$  
     \\  
    s-t cut & $\frac{m}{1 + (m-1)(1 - \curv)}$ & $O(\frac{\sqrt{m} \log m}{1 + (\log m \sqrt{m} - 1)(1 - \curv)})$ 
    %& $\frac{n}{2 + (n - 2)(1 - \curv)}$ 
    & $\tilde{\Omega}(\frac{\sqrt{n}}{1 + (\sqrt{n} - 1)(1 - \curv)})$ %& $\Omega(\sqrt{n/\log n})$  
    \\  
   \hline
  \end{tabular}
  \caption{Summary of our results for constrained minimization (Problem~\ref{prob:min}).}
  % In addition, using a different technique, we provide another bounds for s-t cuts in Section~\ref{consminsec} which is not listed here.} % Abbreviations are described in the main text.}
  \label{tab:results}
\end{table}

%%%%%%%%%%%%%%%%%%%%%%%%%%%%%%%%%%%%%%%%%%%%%%%%%%%%%%%%%%%%%%
\section{Approximating submodular functions everywhere}
\label{learnimproved}

We first address improved bounds for the problem of approximating a monotone submodular function everywhere.
Previous work established 
$\alpha$-approximations $g$ to a submodular function $f$ satisfying $g(S) \leq f(S) \leq \alpha g(S)$ for all $S \subseteq V$ \cite{goemans2009approximating}.
We begin with a theorem showing how any algorithm computing such an approximation may be used to obtain a curvature-specific, improved approximation. Note that the curvature of a monotone submodular function can be obtained within $2n+1$ queries to $f$.
The key idea of Theorem~\ref{thm:f_and_g} is to only approximate the curved part of $f$, and to retain the modular part exactly.\notarxiv{ The full proof is in~\cite{nips2013extendedvcurv}.}
% To following theorem provides a generic transformation that leads from a curvature-independent approximation bound to a bound that respects curvature:%We assume that $f$ is a monotone non-decreasing and non-negative submodular function with a curvature $\curv$.  We then construct the function $f^{\kappa}$, using the construction shown in section~\ref{polycons}. The following theorem shows that we can approximate a submodular function, with curvature $\curv < 1$, within a factor of $\frac{1}{1 - \curv}$.

\begin{theorem}\label{thm:f_and_g}
Given a polymatroid function $f$ with $\curv < 1$, let $f^{\kappa}$ be its \truncated{} version defined in Equation~\eqref{eq:defineg}, and let $\hat{f}^{\kappa}$ be a submodular function satisfying $\hat{f}^{\kappa}(X) \leq f^{\kappa}(X) \leq \alpha(n) \hat{f}^{\kappa}(X)$,
for some $X \subseteq V$. %\STR{for all?}
%Where $g$ is the truncated polymatroid function corresponding to a polymatroid function $f$. Then for every polymatroid function $f$, the function 
Then the function $\hat{f}(X) \triangleq \curv \hat{f^{\kappa}}(X) + (1 - \curv) \sum_{j \in X} f(j)$ satisfies
\begin{align} \label{maineq}
\hat{f}(X) \leq f(X) \leq \frac{\alpha(n)}{1 + (\alpha(n) - 1) (1 - \curv)} \hat{f}(X) \leq \frac{\hat{f}(X)}{1 - \curv}.
\end{align}
\arxiv{The above inequalities hold, even if we use an upper bound $\bar{\kappa_f}$ instead of the actual curvature $\kappa_f$.}
\end{theorem}
\arxiv{
\begin{proof}
The first inequality follows directly from definitions. 
To show the second inequality, note that
%We show the right hand side as follows. Notice that we have 
$\hat{f^{\kappa}}(X) \geq \frac{f^{\kappa}(X)}{\alpha(n)}$, and therefore
\begin{align}
\frac{f(X)}{\curv \hat{f^{\kappa}}(X) + (1 - \curv) \sum_{j \in X} f(j)}&= \frac{\curv f^{\kappa}(X) + (1 - \curv) \sum_{j \in X} f(j)}{\curv \hat{f^{\kappa}}(X) + (1 - \curv) \sum_{j \in X} f(j)} \\
					&\leq  \frac{\curv f^{\kappa}(X) + (1 - \curv) \sum_{j \in X} f(j)}{\curv \frac{f^{\kappa}}{\alpha(n)} + (1 - \curv) \sum_{j \in X} f(j)} \\
					&= \alpha(n) \frac{\curv f^{\kappa}(X) + (1-  \curv) \sum_{j \in X} f(j)}{\curv f^{\kappa}+ (1 - \curv) \alpha(n) \sum_{j \in X} f(j)} \\
					&= \frac{\alpha(n)}{1 + \frac{(\alpha(n) - 1) (1 - \curv) \sum_{j \in X} f(j)}{\curv f^{\kappa}(X) + (1 - \curv) \sum_{j \in X} f(j)}} \\
					&\leq \frac{\alpha(n)}{1 + (\alpha(n)-1)(1 - \curv)}
\end{align}
The last inequality follows since $\curv f^{\kappa}(X) + (1 - \curv) \sum_{j \in X} f(j) \leq \sum_{j \in X} f(j)$. The other inequalities in Eqn.~\eqref{maineq} follow directly from the definitions.

It is also easy to see that all the above inequalities will hold using an upper bound $\bar{\kappa_f} > \kappa_f$ instead of $\kappa_f$ in the definition of the curve-normalized function. The bound in that case would be, 
\begin{align}
\hat{f}(X) \leq f(X) \leq \frac{\alpha(n)}{1 + (\alpha(n) - 1) (1 - \bar{\curv})} \hat{f}(X) \leq \frac{\hat{f}(X)}{1 - \bar{\curv}}
\end{align}
where, $\hat{f}(X) = \bar{\curv} \hat{f^{\bar{\kappa}}}(X) + (1 - \bar{\curv}) \sum_{j \in X} f(j)$, $\hat{f^{\bar{\kappa}}}$ is an approximation of $f^{\bar{\kappa}}$ satisfying $\hat{f^{\bar{\kappa}}}(X) \leq f^{\kappa}(X) \leq \alpha(n) \hat{f^{\kappa}}(X)$ and,
\begin{align}
f^{\bar{\kappa}}(X) = \frac{f(X) - {(1 - \bar{\curv})} \sum_{j \in X} f(j)}{\bar{\curv}}
\end{align}
\end{proof}}
Theorem~\ref{thm:f_and_g} may be directly applied to tighten recent results on approximating submodular functions everywhere.
An algorithm by \citet{goemans2009approximating} computes an approximation to a polymatroid function $f$ in polynomial time by approximating the submodular polyhedron via an ellipsoid. This approximation (which we call the ellipsoidal approximation) satisfies $\alpha(n) = O(\sqrt{n}\log{n})$, and
%This approximation, moreover, is of 
has the form $\sqrt{w^f(X)}$ for a certain weight vector $w^f$.
%The next corollary shows that a tighter approximation is possible if $\curv < 1$.
%
% We will use Theorem~\ref{thm:f_and_g} to tighten the results in \cite{goemans2009approximating} about approximating a submodular function with a polynomial number of oracle queries.
% %We then recall the result by Goemans \textit{et al}, where they show a construction $\hat{g}$ to approximate any non-negative and non-decreasing monotone submodular function.
\arxivalt{
\begin{theorem}[\cite{goemans2009approximating}]
For any polymatroid rank function $f$, one can compute a weight vector $w^f$ and correspondingly an approximation $\sqrt{w^f(X)}$ via a polynomial number of oracle queries such that $\sqrt{w^f(X)} \leq f(X) \leq O(\sqrt{n}\log{n}) \sqrt{w^f(X)}$. %If $g$ is a matroid rank function, then an approximation with $\hat{g}(X) \leq g(X) \leq \sqrt{(n+1)} \hat{g}(X)$ is possible.
%
%Let $g$ be any polymatroid function, and $\hat{g}(X) = \sqrt{w(X)}$. Then for every function $g$, there exists a $w$, such that $\hat{g}(X) \leq g(X) \leq O(\sqrt{n}\log{n}) \hat{g}(X)$. Further if $g$ is a matroid rank function, then there exists a $w$, such that $\hat{g}(X) \leq g(X) \leq \sqrt{n+1} \hat{g}(X)$.
\end{theorem}
%Their approximation (which we refer to as $f^{ea}(X)$) is of the form of $\sqrt{w^f(X)}$ for a certain vector $w^f$. 
The weights $w^f$ are computed via an ellipsoidal approximation of the submodular polyhedron~\cite{goemans2009approximating}. %In particular, they use a certain interesting property of monotone submodular functions, which enables them to find inner and outer approximations of the submodular polyhedron.
}
{% The main result of Goemans \textit{et al} is to provide an algorithm, which is based on approximating the submodular polyhedron by an ellipsoid. They show that for any polymatroid rank function $f$, one can compute an approximation $f^{ea}$ (the Ellipsoidal Approximation) explicitely via a polynomial number of oracle queries such that $f^{ea}(X) \leq f(X) \leq O(\sqrt{n}\log{n}) f^{ea}(X)$. Further their approximation $f^{ea}(X)$ is of the form of $\sqrt{w^f(X)}$ for a certain vector $w^f$.
}
Corollary~\ref{cor:learn} states that a tighter approximation is possible for functions with $\curv < 1$.
%Finally using the two theorems above, we can provide improved bounds on learning monotone non-decreasing submodular functions with a linearity coefficient $\curv > 0$.
\begin{corollary}\label{cor:learn}
Let $f$ be a polymatroid function with $\curv < 1$, and 
let $\sqrt{w^{f^{\kappa}}(X)}$ be the ellipsoidal approximation to the $\kappa$-\truncated{} version $f^{\kappa}(X)$ of $f$. Then
%Given a polymatroid function $f$ with $\curv < 1$, 
the function $f^{ea}(X) = \curv \sqrt{w^{f^{\kappa}}(X)} + (1 - \curv)\sum_{j \in X} f(j)$
satisfies
\begin{align}
\label{eqn:eabound}
f^{ea}(X) \leq f(X) %\leq O(\frac{\sqrt{n} \log{n}}{1 + (\sqrt{n} \log{n} - 1)(1 - \curv(X))}) f^{cea}(X) 
\leq O\left(\frac{\sqrt{n} \log{n}}{1 + (\sqrt{n} \log{n} - 1)(1 - \curv)}\right) f^{ea}(X).
\end{align} 
\end{corollary}
If $\curv=0$, then the approximation is exact. This is not surprising
since a modular function can be inferred exactly within $O(n)$ oracle
calls.\arxiv{
\begin{proof}
  To compute $f^{ea}$, construct the function $f^{\kappa}$ as in Equation~\eqref{eq:defineg}, and apply the algorithm in \cite{goemans2009approximating} to construct the approximation $\sqrt{w^{f^{\kappa}}(X)}$ such that 
%
%The proof of this follows directly from the fact that the algorithm by Goemans \textit{et al}~\cite{goemans2009approximating} produces a submodular function $\hat{g}$ such that 
$\sqrt{w^{f^{\kappa}}(X)} \leq f^{\kappa}(X) \leq O(\sqrt{n} \log n) \sqrt{w^{f^{\kappa}}(X)}$. Note that $\sqrt{w^{f^{\kappa}}(X)}$ is an approximation of $f^{\kappa}$ and not $f$. Then define $f^{ea}(X) \triangleq \curv(X) f^{ea}(X) + (1 - \curv) \sum_{j \in X} f(j)$
%, for any monotone submodular function. Using this fact with the theorem above, we obtain the above corollary.
\end{proof}
}
The following lower bound\notarxiv{ (proved in \cite{nips2013extendedvcurv})} shows that Corollary~\ref{cor:learn} is tight up to logarithmic factors. It refines the lower bound in \cite{goemans2009approximating} to include $\curv$.
% \notarxiv{ The proof of this result is in~\cite{nips2013extendedvcurv}.\looseness-1}
%an analogous to the general lower bound proved in \cite{goemans2009approximating}, but refined to include the curvature coefficient.
%We now show a lower bound on the hardness which is quite close to our bounds above.
\begin{theorem}\label{learnhardness}
Given a submodular function $f$ with curvature $\curv$, there does not exist a (possibly randomized) polynomial-time algorithm that computes an approximation to $f$ within a factor of $\frac{n^{1/2-\epsilon}}{1 + (n^{1/2-\epsilon}-1)(1 - \curv)}$, for any $\epsilon > 0$.
%$o(\frac{\sqrt{n /\log{n}}}{1 + (\sqrt{n/ \log{n}} - 1)\curv})$.
\end{theorem}
\arxiv{\begin{proof}
The information-theoretic proof uses a construction and argumentation similar to that in \cite{goemans2009approximating,svitkina2008submodular}, but perturbs the functions to have the desired curvature.

%  We use the same construction and distinguishability argument as in \cite{goemans2009approximating,svitkina2008submodular}, but use a convex combination of their functions with a cardinality.

%The proof of this theorem is based on the ideas of section~\ref{hardnessframework} and using some of the ideas in~\cite{goemans2009approximating}. 
In the following let $\curv = \kappa$.
Define two monotone submodular functions $h^{\kappa}(X) = \kappa \min\{|X|, \alpha\} + (1 - \kappa) |X|$ and $f_{\kappa}^R(X) = \kappa \min\{\beta + |X \cap \bar{R}|, |X \cap R|, \alpha\} + (1 - \kappa)|X|$, where $R \subseteq V$ is %chosen uniformly at random and has
% being 
a random set of 
cardinality $\alpha$. Let $\alpha$ and $\beta$ be an integer such that $\alpha = x\sqrt{n}/5$ and $\beta = x^2/5$ for an $x^2 = \omega(\log n)$. Both $h^{\kappa}$ and $f_{\kappa}^R$ have curvature equal to $\curv = \kappa$.

Using a Chernoff bound, one can then show that any algorithm that uses a polynomial number of queries can distinguish $h^{\kappa}$ and $f_{\kappa}^R$ with probability only $n^{-\omega(1)}$, and therefore cannot reliably distinguish the functions with a polynomial number of queries~\cite{svitkina2008submodular}.

Therefore, any such algorithm will, with high probability, approximate $h^{\kappa}$ and $f_{\kappa}^R$ by the same function $\hat{f}$. Since the approximation must hold for both functions, the approximation factor must satisfy $h^{\kappa}(R) \leq \gamma \hat{f}(R) \leq \gamma f_{\kappa}^R(R)$, and is therefore lower bounded by $h^{\kappa}(R)/f_{\kappa}^R(R)$. Given an arbitrary $\epsilon>0$, set
%We assume that 
$x^2 = n^{2\epsilon} = \omega(\log n)$. % for an arbitrary $\epsilon > 0$.
Then
\begin{align}
  \frac{h^{\kappa}(R)}{f_{\kappa}^R(R)} &= \frac{\alpha}{(1 - \kappa) \alpha + \kappa \beta} \\
  &= \frac{n^{1/2+\epsilon}}{(1 - \kappa) n^{1/2+\epsilon} + \kappa n^{2\epsilon}} \\
  &= \frac{n^{1/2-\epsilon}}{1 + (n^{1/2-\epsilon}-1)(1 - \kappa)}.
  %O(\frac{\sqrt{n/\log n}}{(\sqrt{n/\log n}- 1)\curv + 1)
\end{align}
%$h(R)/f_R(R) = $. 
Assume there was an algorithm that generates an approximation $\hat{f}'$ with approximation factor $\gamma' < \gamma$. This would imply that $h^{\kappa}(R) / f_{\kappa}^R(R) < \gamma'$, but this contradicts the above derivation.
\end{proof}}
The simplest alternative approximation to $f$ one might conceive is
the modular function $\hat{f}^m(X) \triangleq \sum_{j \in X}f(j)$ which
can easily be computed by querying the $n$ values $f(j)$.
%We end this section, by investigating yet another approximation. This approximation is the simple modular upper bound approximation $\sum_{j \in X} f(j)$, which is extremely simple to obtain. However we can show interesting theoretical properties concerning this approximation. 
\begin{lemma}\label{modapproxlemma}
Given a monotone submodular function $f$, it holds that\notarxiv{\footnote{In \cite{nips2013extendedvcurv}, we show this result with a stronger notion of curvature: $\hat{\kappa_f}(X) = 1 - \frac{\sum_{j \in X} f(j | X \backslash j)}{\sum_{j \in X} f(j)}$.}}
% \STR{This lemma is of a different form that the others (usually $f$ is sandwiched in between the approximations). Is there a way to make this uniform?}\RTS{There is but that makes it complicated. I have added a sentence below.}
\begin{align}\label{modapproxeqn}
f(X) \leq 
\hat{f}^m(X) =
\sum_{j \in X} f(j) \leq \frac{|X|}{1 + (|X| - 1)(1 - \curv(X))} f(X)
\end{align}
\arxiv{
Moreover, it also holds that,
\begin{align}\label{modapproxeqnstronger}
f(X) \leq 
\hat{f}^m(X) =
\sum_{j \in X} f(j) \leq \frac{|X|}{1 + (|X| - 1)(1 - \hat{\curv}(X))} f(X)
\end{align}
}
\end{lemma}
%\notarxiv{ The proof of this Lemma is in~\cite{nips2013extendedvcurv}.}
\arxiv{
\begin{proof}
We first show the result for $\hat{\kappa_f}(X)$, and since it is a stronger notion of curvature, the bound will hold for $\kappa_f(X)$ as well.

We shall use the following facts, which follow from the definitions of submodularity and curvature.
\begin{align}
&\mbox{Fact 1: }(1 - \hat{\curv}(X))\sum_{j \in X} f(j) = \sum_{j \in X} f(j | X \backslash j), \\
&\mbox{Fact 2: }f(X) - f(k) \geq \sum_{j \in X \backslash k} f(j | X \backslash j), \forall k \in X.
\end{align}
%The above facts by the definition of submodularity and curvature. 
%Combining the above, we get that
%\begin{align}
%f(X) \geq f(k) + (1 - \curv(X))\sum_{j \in X \backslash k} f(j), \forall k \in X.
%\end{align}
%Adding all the expressions above for all $k \in X$ provides the above result.
Sum the expressions from Fact 2, $\forall k \in X$, use Fact 1, and we obtain the following series of inequalities,
\begin{align*}
|X| f(X) - \sum_{k \in X} f(k) &\geq \sum_{k \in X} \sum_{j \in X \backslash k} f(j | X \backslash j) \\
	&\geq \sum_{k \in X} \sum_{j \in X} f(j | X \backslash j) - \sum_{k \in X} f(j | X \backslash k) \\
	&\geq (|X| - 1) \sum_{j \in X \backslash k} f(j | X \backslash j) \\
	&\geq (|X| - 1)(1 - \hat{\curv}(X)) \sum_{k \in X} f(k)
\end{align*}
Hence we obtain that,
\begin{align*}
\sum_{k \in X} f(k) \leq \frac{|X|}{1 + (|X|-  1)(1 - \hat{\curv}(X))} f(X)
\end{align*}
From the fact that $1 - \hat{\curv}(X) \geq 1 - \curv(X)$, it immediately follows that,
\begin{align*}
\sum_{k \in X} f(k) \leq \frac{|X|}{1 + (|X|-  1)(1 - \curv(X))} f(X)
\end{align*}
\end{proof}
}The form of Lemma~\ref{modapproxlemma} is slightly different from 
Corollary~\ref{cor:learn}. However, there is a straightforward correspondence: 
 given $\hat{f}$ such that $\hat{f}(X)
\leq f(X) \leq \alpha^{\prime}(n) \hat{f}(X)$, by defining
$\hat{f^{\prime}}(X) = \alpha^{\prime}(n) \hat{f}(X)$, we get that
$f(X) \leq \hat{f^{\prime}}(X) \leq \alpha^{\prime}(n)
f(X)$. 
%This shows that the bounds in
%Corollary~\ref{cor:learn} can essentially be transformed into the form
%of Lemma~\ref{modapproxlemma}.  
\JTR{The sentence above doesn't seem
  to work as the sandwich bound since it is not of the form that $\hat
  f'(X)$ is bounded by two factors involving $f(X)$. I.e., the two
  forms of bound are either: $f$ is sandwiched between two
  approximations both involving $f_\text{approx}$ as in
  Eqn.\eqref{eqn:eabound} or that $f_\text{approx}$ is sandwiched by
  two items involving $f$ as in Eqn.\eqref{modapproxeqn}. What you've
  got is the upper bound on $\hat f'$ involves $\hat f'$ and the lower
  bound involves $f$.  }\RTJ{Sorry that was a typo above. What I meant and which is also true, is the corrected version above. thanks!} Lemma~\ref{modapproxlemma} for the modular
approximation is complementary to 
% offers an interesting contrast to
Corollary~\ref{cor:learn}: First, the modular approximation is better
whenever $|X| \leq \sqrt{n}$.  \JTR{maybe that's only because the EA
  is so bad for average case.}\RTJ{Possibly.}  Second, the bound in
Lemma~\ref{modapproxlemma} depends on the curvature $\curv(X)$ with
respect to the set $X$, which is stronger than $\curv$. Third,
$\hat{f}^m$ is extremely simple to compute. For sets of larger
cardinality, however, the ellipsoidal approximation of
Corollary~\ref{cor:learn} provides a better approximation, in fact,
the best possible one (Theorem~\ref{learnhardness}).\JTR{add cite.}\RTJ{done}  In a similar manner,
Lemma~\ref{modapproxlemma} is tight for any modular approximation to a
submodular function: %\notarxiv{ (again proved in~\cite{nips2013extendedvcurv})}:\looseness-1
\begin{lemma}\label{lem:lb_linear}
  For any $\kappa > 0$, there exists a monotone submodular function $f$ with curvature $\kappa$ such that no modular upper bound on $f$\arxiv{,$\hat{f}(X) = \sum_{j \in X}w(j) \geq f(X), \forall X \subseteq V$,} can approximate $f(X)$ to a factor better than $\frac{|X|}{1 + (|X| - 1)(1 - \kappa_f)}$.
\JTR{clarify this lemma, since the bound here is $X$ dependent, so say
that it depends on the particular size of the argument, etc.}\RTJ{I think the above should clarify it.}\STR{Why not just replace $|X|$ by $n$?}
\end{lemma}
\arxiv{\begin{proof}  Let 
  $f^{\kappa}(X) = \kappa \min\{|X|, 1\} + (1 - \kappa)|X|$. Then $f^{\kappa}(X) = \kappa + (1 - \kappa)|X| = 1 + (1 - \kappa)(|X| - 1)$ for all $\emptyset \subset X \subseteq V$. Since $\hat{f}^m$ is an upper bound, it must satisfy $\hat{f}(j) = w(j) \geq 1$ for all $j \in V$. Therefore, $\hat{f}^m(X) = |X|, X \neq \emptyset$.
\end{proof}}
The improved curvature dependent bounds immediately imply better bounds for the class of concave over modular functions used in 
% This fact is of practical relevance, since this class of functions have been used in a number of machine learning applications~
\cite{lin2011-class-submod-sum,jegelka2011-nonsubmod-vision, rkiyeruai2012}.
\begin{corollary}\label{concvmodres}
Given weight vectors $w_1, \cdots, w_k \geq 0$, and a submodular function $f(X) = \sum_{i = 1}^k \lambda_i [w_i(X)]^a, \lambda_i \geq 0$, for $a \in (0, 1)$, it holds that  $f(X) \leq \sum_{j \in X} f(j) \leq |X|^{1-a} f(X)$
\end{corollary}
\arxiv{
\begin{proof}
We first show this result independent of curvature, and then show how the curvature dependent bound also implies this improved bound. First define $f(X) = [w(X)]^a$, for $a \in (0, 1]$ and $w \geq 0$. since $x^a$ is a concave function for $a \in (0, 1]$, we have from Jensen's inequality that, given $x_1, x_2, \cdots, x_n \geq 0$,
\begin{align*}
\frac{\sum_{i = 1}^n x_i^a}{n} \leq \left ( \frac{\sum_{i = 1}^n x_i }{n} \right )^a
\end{align*}
Notice that, when $f(X) = [w(X)]^a$, we have that $f(i) = w(i)^a$. Hence $\sum_{j \in X} f(i) = \sum_{i \in X} w(i)^a$. Hence from the inequality above, it directly holds that,
\begin{align*}
\frac{\sum_{i \in X} w(i)^a}{|X|} \leq  \frac{[\sum_{i \in X} w(i)]^a }{|X|^a}
\end{align*}
and hence, $\sum_{j \in X} f(j) \leq |X|^{1 - a} f(X)$. This inequality also holds for a sum of concave over modular functions, since for each $w_i \geq 0$, we have 
\begin{align}\label{jeneqcv}
\sum_{j \in X} w_i(j)^a \leq |X|^{1 - a} w_i(X)^a.
\end{align} 
Moreover, when $f(X) = \sum_{i = 1}^k \lambda_i [w_i(X)]^a$, the modular upper bound $\sum_{j \in X} f(j) = \sum_{j \in X} \sum_{i = 1}^k \lambda_i [w_i(j)]^a$. Summing up eqn.~\eqref{jeneqcv} for all $i$, we have that,
\begin{align}\label{jeneqcv}
\sum_{i = 1}^k \sum_{j \in X} \lambda_i w_i(j)^a \leq |X|^{1 - a} \sum_{i = 1}^k \lambda_i w_i(X)^a.
\end{align} 
\end{proof}
We next show that this result can also be seen from the curvature of the function.
\begin{lemma}
Given weight vectors $w_1, \cdots, w_k \geq 0$, and a submodular function $f(X) = \sum_{i = 1}^k \lambda_i [w_i(X)]^a, \lambda_i \geq 0$, for $a \in (0, 1]$, it holds that, 
\begin{align}
\hat{\kappa_f}(X) \leq 1 - \frac{a}{|X|^{1-a}}
\end{align}
\end{lemma}
\begin{proof}
Again, let $f(X) = [w(X)]^a$, for $w \geq 0$ and $a \in (0, 1]$. Then,
\begin{align*}
f(j | X \backslash j) &= f(X) - f(X \backslash j) \\
		&= [w(X)]^a - [w(X) - w(j)]^a \\
		&\geq \frac{a w(j)}{w(X)^{1-a}} 
\end{align*}
The last inequality again holds due to concavity of $g(x) = x^a$. In particular, for a concave function, $g(y) - g(x) \leq g^{\prime}(x) (y - x)$, where $g^{\prime}$ is the derivative of $g$. Hence $g(x) \geq g(y) + g^{\prime}(x) (x - y)$. Substitute $y = w(X) - w(j), x = w(X)$ and $g(x) = x^a$, and we get the above expression.

Hence we have,
\begin{align*}
1 - \hat{\curv}(X) &= \frac{\sum_{j \in X} f(j | X \backslash j)}{\sum_{j \in X} f(j)} \\	
		&\geq \frac{ a \sum_{j \in X} w(j) w(X)^{a-1}}{\sum_{j \in X} w(j)^a} \\
		&\geq \frac{a w(X)^a}{\sum_{j \in X} w(j)^a} \\
		&\geq \frac{a}{|X|^{1-a}}
\end{align*}
The last inequality follows from the previous Lemma.
\end{proof}
Hence from the curvature dependent bound, we obtain a slightly weaker bound, which still gives a $O(|X|^{1 - a})$ bound for the modular upper bound.
\begin{align*}
\sum_{j \in X} f(j) &\leq \frac{1}{1 - \hat{\curv}(X)} f(X) \\
	&\leq \frac{|X|^{1 - a}}{a} f(X) \\
	&\leq O(|X|^{1 - a}) f(X)
\end{align*}
}
In particular, when $a = 1/2$, the modular upper bound approximates the sum of square-root over modular functions by a factor of $\sqrt{|X|}$.

\section{Learning Submodular functions}
We next address the problem of learning submodular functions in a
PMAC setting \cite{balcanlearning}.
%consider the problem introduced in~\cite{balcanlearning}, of PMAC (Probably Mostly Approximately Correct) learn an unknown but fixed submodular function $f$. 
The PMAC (Probably Mostly Approximately Correct) framework is an
extension of the PAC framework~\cite{valiant1984theory} to allow
multiplicative errors in the function values from a fixed but unknown
distribution $\mathcal D$ over $2^V$.  We are given training
samples $\{(X_i, f(X_i)\}_{i=1}^m$ drawn i.i.d.\ from 
$\mathcal D$.  The algorithm may take time
polynomial in $n$, $1/\epsilon$, $1/\delta$ to compute a (polynomially-representable) function $\hat{f}$ that is a good approximation to $f$
with respect to $\mathcal D$. Formally, $\hat{f}$ must satisfy that
%computation over these, and provide a submodular function, which with high probability is a good approximation of $f$ over $\mathcal D$. Formally this is expressed as:
\begin{align}
\mathrm{Pr}_{X_1, X_2, \cdots, X_m \sim \mathcal D}\left[\mathrm{Pr}_{X \in \mathcal D}[ \hat{f}(X) \leq f(X) \leq \alpha(n) \hat{f}(X)] \geq 1 - \epsilon \right] \geq  1 - \delta %\nonumber
\end{align}
for some approximation factor $\alpha(n)$. 
%Further we would expect the algorithm to be polynomial time in $n, 1/\epsilon, 1/\delta$. 
\citet{balcanlearning} propose an algorithm that PMAC-learns any
monotone, nonnegative submodular function within a factor $\alpha(n) =
\sqrt{n+1}$ by reducing the problem to that of learning a binary
classifier. 
If we assume that we have an upper bound on the curvature $\kappa_f$, or that we can estimate it \footnote{note that $\kappa_f$ can be estimated from a set of $2n+1$ samples $\{(j,f(j))\}_{j \in V}$, $\{(V,f(V))\}$, and $\{(V \backslash j, f(V \backslash j)\}_{j \in V}$ included in the training samples}, and have access to the value of the singletons $f(j), j \in V$, then we can obtain better learning results with non-maximal curvature:
%We also denote $\mathcal D^{\prime}$ as a set of $2n+1$
%additional samples in the form of the singletons $\{(j,f(j))\}_{j \in
%  V}$, the ground set $\{(V,f(V))\}$, and $\{(V \backslash j, f(V \backslash
%j)\}_{j \in V}$. If we make the further assumption that the training
%samples includes $\mathcal D^{\prime}$ (or equivalently that we know $\kappa_f$ or have an upper bound on it, and the valuations of the singletons), Theorem~\ref{thm:f_and_g} can
%be used to improve their algorithm to achieve a curvature-dependent
%approximation factor.\looseness-1
%the authors provide a $\sqrt{n+1}$ approximation algorithm. The main idea of their algorithm is to reduce this problem into one of learning a binary classifier from i.i.d samples in a passive and supervised learning setting. We use their algorithm to provide an improved curvature dependent bound below.
\begin{lemma}\label{thm:pmac}
\STR{Formulating this as $f$ having a known upper bound on the curvature instead of having exactly known curvature -- does that still go through?}\RTS{The proof will go through if we have an upper bound on the curvature. We still need the singletons though)}
Let $f$ be a monotone submodular function for which we know an upper bound on its curvature and the singleton weights $f(j)$ for all $j \in V$. 
% (or a known upper bound) $\curv$, then 
For every $\epsilon, \delta > 0$ there is an algorithm that uses a polynomial number of training examples, runs in time polynomial in $(n, 1/\epsilon, 1/\delta)$ and PMAC-learns $f$ within a factor of $\frac{\sqrt{n+1}}{1 + (\sqrt{n+1} - 1)(1 - \curv)}$.
If $\mathcal D$ is a product distribution, then there exists an algorithm that PMAC-learns $f$ within a factor of $O(\frac{\log \frac{1}{\epsilon}}{1 + (\log \frac{1}{\epsilon} - 1)(1 - \curv)})$.
%Given a class of monotone non-negative submodular function with curvature coefficient $\curv$, there is an algorithm which for every $\epsilon, \delta > 0$, PMAC learns $f$ upto a factor of $\frac{\sqrt{n+1}}{1 + (\sqrt{n+1} - 1)(1 - \curv)}$, using polynomial number of training examples and runs in time polynomial in $(n, 1/\epsilon, 1/\delta)$. Under the special case when $\mathcal D$ is a product distribution, there exists an algorithm which PMAC learns $f$ upto a factor of $O(\frac{\log \frac{1}{\epsilon}}{1 + (\log \frac{1}{\epsilon} - 1)(1 - \curv)})$.
\end{lemma}
The algorithm of Lemma~\ref{thm:pmac} uses the reduction of \citet{balcanlearning} to learn the $\curv$-\truncated{} version $f^{\kappa}$ of $f$.  
From the learned function
%Denote this learned function as 
$\hat{f^{\kappa}}(X)$, we construct the final estimate $\hat{f}(X) \triangleq \curv \hat{f^{\kappa}}(X) + (1 - \curv) \sum_{j \in X} f(j)$. Theorem~\ref{thm:f_and_g} implies 
Lemma~\ref{thm:pmac} for this $\hat{f}(X)$.
% \STR{Don't we need to learn the modular part too? How do we achieve that?}
%The main idea here is to apply the algorithm of Balcan \& Harvey~\cite{balcanlearning} to learn the $\curv$-\truncated{} version $f^{\kappa}$ of $f$. Denote this learnt function as $\hat{f^{\kappa}}(X)$. Then defining $\hat{f}(X) \triangleq \curv \hat{f^{\kappa}}(X) + (1 - \curv) \sum_{j \in X} f(j)$, ensures that $\hat{f}$ has the above guarantees in the learning model.
\arxiv{\begin{proof}
The proof of this theorem directly follows from the results in~\cite{balcanlearning} and those from section~\ref{learnimproved}. The idea is that, we use the PMAC setting and algorithm from~\cite{balcanlearning}. 
%Then let $X$ be an instance which meets the required bound (i.e it is a ``good`` approximation of $f$).
We use the same construction as section~\ref{learnimproved}, and construct the function $f^{\kappa}(X)$ which is the \truncated{} version of $f$. Let $\hat{f^{\kappa}}(X)$ be the function learn from $f^{\kappa}$ using the algorithm from~\cite{balcanlearning}. Then define $\hat{f}(X) = (1 - \curv)\hat{f^{\kappa}}(X) + \curv \sum_{j \in X} f(j)$ and an analysis similar to that in section~\ref{learnimproved} conveys that the function $\hat{f}$ is within a factor of $\frac{\sqrt{n+1}}{1 + (\sqrt{n+1} - 1)(1 - \curv)}$. Note that moreover, whenever the bound $\hat{f^{\kappa}}(X) \leq f^{\kappa}(X) \leq \sqrt{n+1} \hat{f^{\kappa}}(X)$, the above curvature dependent bound will also hold. Hence the curvature dependent bound holds with high probability on a large measure of sets. 
The case for product distributions also follows from very similar lines and the results from~\cite{balcanlearning}.
\end{proof}}
\arxiv{Lemma~\ref{thm:pmac} is tight as we show below.} \notarxiv{Moreover, 
%we also show in~\cite{nips2013extendedvcurv}, that there does not exist a 
no polynomial-time algorithm can be guaranteed to PMAC-learn $f$ 
%for every $\epsilon, \delta > 0$, 
within a factor of $\frac{n^{1/3- \epsilon^{\prime}} }{1 + (n^{1/3 - \epsilon^{\prime}} - 1)(1 - \curv)}$, for any $\epsilon^{\prime} > 0$~\cite{nips2013extendedvcurv}.}
% We show a lower bound below, again using an analysis very similar to~\cite{balcanlearning}. 
\arxiv{\begin{lemma}
Given a class of submodular functions with curvature $\curv$, there does not exist a polynomial-time algorithm (which possibly even has information about $\curv$) that is guaranteed to PMAC-learn $f$ for every $\epsilon, \delta > 0$ within a factor of $\frac{n^{1/3- \epsilon^{\prime}} }{1 + (n^{1/3 - \epsilon^{\prime}} - 1)(1 - \curv)}$, for any $\epsilon^{\prime} > 0$.
\end{lemma}}
\arxiv{\begin{proof}
Again, we use the same matroid functions used in~\cite{balcanlearning}. Notice that the construction of~\cite{balcanlearning}, provides a family of matroids and a collection of sets $\mathcal B$, with $|A| = n^{1/3}$, such that $f(A) = |A|, A \in \mathcal B$ and $f(A) = \beta = \omega(\log n), A \notin \mathcal B$. Again set $\beta = n^{\epsilon^{\prime}}$, and using a analysis and construction similar to the hardness proof of section~\ref{learnimproved} and Theorem 9 from~\cite{balcanlearning} conveys that the lower bound for this problem is $\tilde{\Omega}(\frac{n^{1/3} }{1 + (n^{1/3} - 1)(1 - \curv)})$
\end{proof}}
%In the results above, we assume that the training examples include the additional singletons 
%$\mathcal D^{\prime}$ 
%or that they can be added easily into the training examples. 
%\arxiv{This may not be a bad assumption in practice. Note, however, that the analysis above goes through even if we have an upper bound on the curvature (which is often the case, if we know the class of the submodular function), and the singletons $\{j, f(j)\}_{j \in V}$. 
%}
We end this section by showing how we can learn with a construction analogous to that in Lemma~\ref{modapproxlemma}.
% obtain a bound with a tighter version of curvature, and without the above assumption. The dependence on the ground-set size $n$, however, is weaker. 
\begin{lemma}\label{thm:pmac2}
If $f$ is a monotone submodular function with known curvature (or a known upper bound) $\hat{\curv}(X), \forall X \subseteq V$, then for every $\epsilon, \delta > 0$ there is an algorithm that uses a polynomial number of training examples, runs in time polynomial in $(n, 1/\epsilon, 1/\delta)$ and PMAC learns $f(X)$ within a factor of $1 + \frac{|X|}{1 + (|X| - 1)(1 - \hat{\curv}(X))}$.
\end{lemma}
\arxivalt{Before proving this result, we compare}{Compare} this result to Lemma~\ref{thm:pmac}. 
Lemma~\ref{thm:pmac2} leads to better bounds for small sets, whereas Lemma~\ref{thm:pmac} provides a better general bound.
Moreover, in contrast to Lemma~\ref{thm:pmac}, here we only need an upper bound on the curvature and do not need to know the singleton weights $\{f(j), j \in V\}$. Note also that, while $\kappa_f$ itself is an upper bound of $\hat{\kappa_f}(X)$, often one does have an upper bound on $\hat{\kappa_f}(X)$ if one knows the function class of $f$ (for example, say concave over modular). In particular, an immediate corollary is that the class of concave over modular functions $f(X) = \sum_{i = 1}^k \lambda_i [w_i(X)]^a, \lambda_i \geq 0$, for $a \in (0, 1)$ can be learnt within a factor of $\min\{\sqrt{n+1}, 1 + |X|^{1-a}\}$.
% Firstly, notice that this algorithm is in the worst case a factor $n$, while the above one is $\sqrt{n}$. On the other hand, this theoretically has tighter guarantees for sets with smaller cardinality (even if we ignore the curvature part). Moreover, we do not need to exactly know $\hat{\curv}(S)$, but it suffices to know an upper bound on it (which is often realistic, if we know the function class). On the other hand, while the algorithm of Lemma~\ref{thm:pmac} requires the singletons $\{j, f(j)\}_{j \in V}$ along with an estimate (or bound) on the curvature, this technique just requires a bound on the curvature.
\arxiv{\begin{proof}
To prove this result, we adapt Algorithm 2 in~\cite{balcanlearning} to curvature and modular approximations. Following their arguments,
%Similar to~\cite{balcanH11}, 
we reduce the problem of learning a submodular function to that of learning a linear seperator,
while separately handling the subset of instances where $f$ is zero. % and non-zero respectively.
\STR{Rewriting this because I do not understand it. Also, $\mathcal{D}$ is the \emph{distribution} and not the training sample.} We detail the parts where our proof deviates from \citep{balcanlearning}.

We divide $2^V$ into the support set $\mathcal{S} = \{X \subseteq V \mid f(X) > 0\}$ of $f$ and its complement $\mathcal{Z} = \{X \subseteq V \mid f(X) = 0\}$.
Using samples from $\mathcal{D}'$, we generate new, binary labeled samples from a distribution $\mathcal{D}'$ on $\{0,1\}^n \times \mathcal{R}$ that will be used to learn the linear separator. 
These samples differ slightly from those in \cite{balcanlearning}. Let
\begin{align}
\alpha(X) = \frac{|X|}{1 + (|X| - 1)(1 - \hat{\kappa_f}(X))}.
\end{align}
To sample from $\mathcal{D}'$, we repeatedly sample from $\mathcal{D}$ until we obtain a set $X \in \mathcal{S}$. For each such $X$, we flip a fair coin and, with equal probability, generate a sample point from $\mathcal{D}'$ as
\begin{align}
  x = (1_X, f(X)) &\text{ and label } y=+1 \quad \text{ if heads}\\
  x = (1_X, (\alpha(X)+1)f(X)) &\text{ and label } y=-1 \quad \text{ if tails.}  
\end{align}
We observe that the generated positive and negative sample are linearly separable with the separator $u = (w, -1)$, where $w(j)=0$ if $f(j)=0$, and $w(j) = f(j) + \delta$ if $f(j)>0$, with $\delta$ such that $0 < \delta < \min_{j \in \mathcal S} f(X_j)/n$:
\begin{align}
  u^\top (1_X, f(X)) &= \sum_{j \in X} f(j) + \delta |X| - f(X) > 0\\
  u^\top (1_X, (\alpha(X)+1)f(X)) &= \sum_{j \in X} f(j) + \delta |X| - (\alpha(X)+1)f(X) < 0
\end{align}
for all $X \subseteq V$. The second inequality holds since $\sum_{j \in X} f(j) \leq \alpha(X) f(X)$ and $\delta |X| \leq \delta n < f(X)$. (For points in $\mathcal{Z}$, we have that $u^\top(1_X, f(X)) = 0$.)

The final algorithm generates a sample from $\mathcal{D}'$ for each sample $X \in \mathcal{S}$ from $\mathcal{D}$. For each $X \in \mathcal{Z}$, it adds the constraint that $w(j) = 0$ for all $j \in X$.
We then find a linear separator $u = (w,-z)$ and output the function $\hat{f}(X) \triangleq w(X)/z$. This is possible by the above arguments.

This function satisfies the approximation constraints for the set $\mathcal{Y}$ of all training points $X \in \mathcal{S}$ for which both generated samples are labeled correctly: the correct labelings $w(X) - z f(X) > 0$ and $w(X) - z(\alpha(X) + 1) f(X) < 0$ imply that
\begin{align}
f(X) \leq \hat{f(X)} = \frac{w(X)}{z} \leq (\alpha(X) + 1) f(X).
\end{align}
Similarly, the constraints on $w$ imply that the same holds for any subset of the union of the training samples in $\mathcal{Z}$.

It remains to show that for sufficiently many, i.e., $\ell = \frac{16n}{\epsilon}\log(\frac{n}{\delta\epsilon})$ samples, the sets $\mathcal{S}\setminus\mathcal{Y}$ and $\mathcal{Z}\setminus \bigcup_{X_i \in \mathcal{Z}, i \leq \ell}X_i$ have small (at most $1-2\epsilon$) measure. This follows from Claim 5 in \citep{balcanlearning}.

\end{proof}
}
%It seems likely, however, that the results above will still hold even if there is a PAC learning algorithm to learn $\mathcal D^{\prime}$. We do not pursue this here, however.\looseness-1
%However the results above go through even if we assume we can PAC learn the instances of $\mathcal D^{\prime}$.  In other words, we do not need to exact examples of valuations in $\mathcal D^{\prime}$ but only a PAC learning algorithm for these restriected sets.
%An interesting problem is to include the estimation of the singleton values from the training data. We leave this as an open problem.\STR{Would be realy good to be able to say something about that. Maybe using Andreas' Fourier analysis? (as a reviewer I would ask about that ;)}
%Notice that in the above, we assume that the algorithm is additionally given the curvature  $\curv$ and  the values of the singletons $f(j)$ of the true submodular function. Both these can be obtained by adding to the training data $2n$ additional examples in the form of singletons $\{j\}, \forall j \in V$ and the sets $\{V \backslash j\}, \forall j \in V$. Typically in applications these can often easily be done. However it might be interesting to consider the problem of estimating these from the data, without needing to explicitely query them. We leave this as an open problem. \looseness-1

\section{Constrained submodular minimization}\label{consminsec}
Next, we apply our results to the minimization of submodular functions under constraints. Most algorithms for constrained minimization use one of two strategies: they apply a convex relaxation~\cite{iwata2009submodular,jegelka2011-inference-gen-graph-cuts}, or they optimize a surrogate function $\hat{f}$ that should approximate $f$ well~\cite{goel2009approximability,goemans2009approximating,jegelka2011-inference-gen-graph-cuts}. We follow the second strategy and propose a new, widely applicable curvature-dependent choice for surrogate functions. A suitable selection of $\hat{f}$ will ensure theoretically optimal results. Throughout this section, we refer to the optimal solution as $X^* \in \argmin_{X \in \mathcal{C}}f(X)$.
%
%We here provide a general framework of algorithms, which rely on minimizing a surrogate function $\hat{f}$, which approximates the given submodular function $f$. Suitable choices of surrogate submodular functions ensure tight results for these problems.  The following theorem provides some necessary conditions for the approximation $\hat{f}$ to obtain curvature dependent guarantees:

%As for Problems~\ref{prob1} and \ref{prob2}, our strategy is to only approximate the curved part $f^{\kappa}$ of a given function $f$, and to accurately retain the modular part. We then optimize the function \STR{How about introducing a special name for the function below and defining it only once? We use the same scheme everywhere, right?}
%\begin{equation}
%  \label{eq:3}
%  \appf(X) \triangleq (1 - \curv) \hat{f^{\kappa}}(X) + \curv \sum_{j \in X} f(j)
%\end{equation}
\begin{lemma}\label{approxguarantee}
Given a submodular function $f$, let $\hat{f}_1$ be an approximation of $f$ such that $\hat{f}_1(X) \leq f(X) \leq \alpha(n) \hat{f}_1(X)$, for all $X \subseteq V$. 
%Also, let $X^* = \argmin_{X \in \mathcal C} f(X)$. 
Then any minimizer
%Define the approximation $$ and 
$\widehat{X}_1 \in \argmin_{X \in \mathcal C} \hat{f}(X)$ of $\hat{f}$ satisfies $f(\widehat{X}) \leq \alpha(n) f(X^*)$. Likewise, if an approximation of $f$ is such that $f(X) \leq \hat{f}_2(X) \leq \alpha(X) f(X)$ for a set-specific factor $\alpha(X)$,
% is a set-specific approximation factor,
then its minimizer $\tilde{X}_2 \in \argmin_{X \in \mathcal C} \hat{f}_2(X)$ satisfies $f(\widehat{X}_2) \leq \alpha(X^*) f(X^*)$. If only $\beta$-approximations\footnote{A $\beta$-approximation algorithm for minimizing a function $g$ finds set $X: g(X) \leq \beta \min_{X \in \mathcal C} g(X)$} are possible for minimizing $\hat{f}_1$ or $\hat{f}_2$ over $\mathcal C$, then the final bounds are $\beta \alpha(n)$ and $\beta \alpha(X^*)$ respectively.
%If $\widehat{X}$ is only a $\beta$-approximate minimizer of $\hat{f}$, then $f(\widehat{X}) \leq \frac{\beta M}{1 + (M - 1) (1 - \curv)} f(X^*)$.
\end{lemma}
\arxiv{\begin{proof}
We prove the first part and the second part similarly follows. Given that,
\begin{align}
\hat{f}(X) \leq f(X) \leq \alpha(n) \hat{f}(X)
\end{align} 
Then, if $\hat{X}$ is the optimal solution for minimizing $\hat{f}$ over $\mathcal C$. We then have that, 
\begin{align}
f(\hat{X}) \leq \alpha(n) \hat{f}(\hat{X}) \leq \alpha(n) \hat{f}(X^*) \leq \alpha(n)  f(X^*)
\end{align} 
where $X^*$ is the optimal solution of $f$. 
\end{proof}}
%\notarxiv{ We prove the above Lemma in~\cite{nips2013extendedvcurv}. }
For Lemma~\ref{approxguarantee} to be practically useful, it is essential that $\hat{f}_1$ and $\hat{f}_2$ be efficiently optimizable over $\mathcal C$. We discuss two general curvature-dependent approximations that work for a large class of combinatorial constraints. In particular, we use Theorem~\ref{thm:f_and_g}: we decompose $f$ into $f^{\kappa}$ and a modular part $f^m$, and then approximate $f^\kappa$ while retaining $f^m$, i.e., $\hat{f} = \hat{f}^\kappa + f^m$.
\STR{We are approximating only the curv-normalized part, right? i.e., what is $\hat{f}^\kappa$ here?? I tried to change that.}
% The key here is  that the approximation $\hat{f}$ must be easy to optimize over $\mathcal C$.  In the sequel, we shall investigate two general approximations which work for a large class of combinatorial constraints. While these approximations provide tight results in most cases, certain specific constraints require certain improved approximations.
%
 % %The first scheme
%relies on the \truncated{} function corresponding to a polymatroid
%function and hence we call it the \truncated{} polymatroidal approach
%(\TRC{}PA), 
The first approach uses a simple modular upper bound (MUB) and the second relies on the Ellipsoidal approximation (EA) we used in Section~\ref{learnimproved}.

\textbf{MUB:} The simplest approximation to a submodular function is the modular approximation $\hat{f}^m(X) \triangleq \sum_{j \in X}f(j) \geq f(X)$. Since here, $\hat{f}^\kappa$ happens to be equivalent to $f^m$, we obtain the overall approximation $\hat{f} = \hat{f}^m$.
Lemmas~\ref{approxguarantee} and~\ref{modapproxlemma} directly imply a set-dependent approximation factor for $\hat{f}^m$:
% In particular here, $\tilde{f}(X) = \sum_{j \in X} f(j)$.
\begin{corollary}\label{SAAcorr}
  % Let $X^* = \argmin_{X \in \mathcal C} f(X)$ and 
Let $\widehat{X} \in \mathcal{C}$ be a $\beta$-approximate solution for minimizing $\sum_{j \in X} f(j)$ over $\mathcal C$, i.e. $\sum_{j \in \widehat{X}} f(j) \leq \beta \min_{X \in \mathcal C} \sum_{j \in X} f(j)$. Then
%Algorithm MUB returns a set $\hat{X}$ such that,
\begin{align}
  f(\hat{X}) \leq \frac{\beta |X^*|}{1 + (|X^*| - 1)(1 - \curv(X^*))}f(X^*).
\end{align}
%where $X^* = \argmin_{X \in \mathcal C} f(X)$ and $\beta$ is the approximation factor of minimizing a modular function over $\mathcal C$.
\end{corollary}
Corollary~\ref{SAAcorr} has also been shown in~\cite{rkiyersemiframework2013}. 
\arxiv{Thanks to Lemma~\ref{modapproxlemma} and the second part of Lemma~\ref{approxguarantee}, however, we can provide a much simpler proof.} 
Similar to the algorithms in \cite{rkiyersemiframework2013, rkiyersubmodBregman2012, iyermirrordescent}, MUB can be extended to an iterative algorithm yielding performance gains in practice. 
In particular, Corollary~\ref{SAAcorr} implies improved approximation bounds for practically relevant concave over modular functions, such as those used in \cite{jegelka2011-nonsubmod-vision}.
For instance, for $f(X) = \sum_{i=1}^k\sqrt{\sum_{j \in X}w_i(j)}$, we obtain a worst-case approximation bound of $\sqrt{|X^*|} \leq \sqrt{n}$. This is significantly better than the worst case factor of $|X^*|$ for general submodular functions.
%Finally, we point out that the curvature dependent bounds convey improved guarantees for the class of concave over modular functions, which have practically been used in applications involving constrained minimization~\cite{jegelka2011-nonsubmod-vision}. In particular, algorithm MUB (or the iterative algorithm from \cite{rkiyersemiframework2013}) admits a worst case approximation factor of $\sqrt{|X^*|} \leq \sqrt{n}$, with the square-root over modular function (Corollary~\ref{concvmodres}). This is significantly better than the worst case factor of $|X^*|$ for general submodular functions.\looseness-1

%%%%%%%%%%%%%%%%%%%% EA %%%%%%%%%%%%%
\textbf{EA: }Instead of employing a modular upper bound, we can approximate $f^{\kappa}$ using the
construction by \citet{goemans2009approximating}, as in Corollary~\ref{cor:learn}. 
%In this case, $\hat{f}(X) = f^{ea}(X)$ from Section~\ref{learnimproved}. 
In that case, $\hat{f}(X) = \kappa_f\sqrt{w^{f^\kappa}(X)} + (1-\curv)f^m(X)$ has a special form: a weighted sum of a concave function and a modular function. Minimizing such a function over constraints $\mathcal{C}$ is harder than minimizing a merely modular function, but with the algorithm in \citep{nikolova2010approximation} we obtain an FPTAS\footnote{The FPTAS will yield a $\beta=(1+\epsilon)$-approximation through an algorithm polynomial in $\frac{1}{\epsilon}$.} for minimizing $\hat{f}$ over $\mathcal{C}$ whenever we can minimize a nonnegative linear function over $\mathcal{C}$.

\begin{corollary}\label{EAcorr}
For a submodular function with curvature $\curv < 1$, 
algorithm EA will return a solution $\widehat{X}$ that satisfies
\begin{equation}
  \label{eq:4}
  f(\widehat{X}) \leq O\left(\frac{\sqrt{n} \log n}{(\sqrt{n} \log n - 1)(1 - \curv) + 1)}\right) f(X^*).
\end{equation}
% where $X^* \in \argmin_{X \in \mathcal{C}}f(X)$.
%is guaranteed to minimize $f$ over $\mathcal C$ upto a factor of $O(\frac{\beta \sqrt{n} \log n}{(\sqrt{n} \log n - 1)(1 - \curv) + 1)})$, where 
\end{corollary}
\arxiv{\begin{proof}
  We use the important result from~\cite{nikolova2010approximation} where they show that any function of the form $\lambda_1 \sqrt{m_1(X)} + \lambda_2 m_2(X)$ where $\lambda_1 \geq 0, \lambda_2 \geq 0$ and $m_1$ and $m_2$ are positive modular functions, has a FPTAS, provided a modular function can easily be optimized over $\mathcal C$. Notice that our function is exactly of that form. Hence $\hat{f}(X)$ can be approximately optimized over $\mathcal C$. This bound then translates into the approximation guarantee using Corollary~\ref{cor:learn} and the first part of Lemma~\ref{approxguarantee}.
  
%It now remains to show that this translates into the approximation guarantee. Notice that from corollary~\ref{cor:learn}, we know that there exists a $\hat{f}$ such that $\hat{f}(X) \leq f(X) \leq \alpha(X) \hat{f}(X), \forall X$ where $\alpha(X) = O(\frac{\sqrt{n} \log n}{(\sqrt{n} \log n - 1)(1 - \curv(X)) + 1)})$. Then, if $\hat{X}$ is the $1 + \epsilon$ approximately optimal solution for minimizing $\hat{f}$ over $\mathcal C$. We then have that $f(\hat{X}) \leq \alpha(\hat{X}) \hat{f}(\hat{X}) \leq \alpha(\hat{X}) (1 + \epsilon) \hat{f}(X^*) \leq \alpha(\hat{X}) (1 + \epsilon) f(X^*)$ where $X^*$ is the optimal solution. We can set $\epsilon$ to any constant, say $1$ and we get the result.
 % and hence every iteration can be approximated up to a factor $(1 + \epsilon)$, and hence we get the above asymptotic result, for a large class of constrained optimization problems.
\end{proof}}

Next, we apply the results of this section to specific optimization problems, for which we show (mostly tight) curvature-dependent upper and lower bounds. \notarxiv{ We just state our main results; a more extensive discussion along with the proofs can be found in~\cite{nips2013extendedvcurv}.}

%We next, demonstrate our framework by providing improved
%approximation guarantees and hardness results for a number of
%constrained submodular minimization problem. In most cases, we provide an improved curvature dependent bounds with matching lower bounds. In particular, we show for the cardinality lower bound constraints, spanning trees, perfect matchings and edge cover constraints, algorithms EA and MUB provide tight bounds. In the case of shortest paths however, we use a slightly more involved approximation, which combines both EA and MUB, in order to get the tight bounds. Similarly in the case of s-t cuts, we use an approximation based on the polymatroidal flows.
%
%\subsection{\Limited{} constraints}
%In this section, we assume that our constraints are $(\kappa,
%P)$-\limited{}. In these cases, we can show that we can achieve near
%tight approximation bounds (with in most cases a matching
%hardness). The main algorithms we analyze here are the \TRC{}PA and
%SAA. In most cases EA provides a bound which is weaker (yet
%asymptotic).
%

\textbf{Cardinality lower bounds (SLB). }
A simple constraint is a lower bound on the cardinality of the solution, i.e., $\mathcal C = \{X \subseteq V: |X| \geq k\}$.
\citet{svitkina2008submodular} prove that for monotone submodular functions of arbitrary curvature, it is impossible to find a polynomial-time algorithm with an approximation factor better than $\sqrt{n/ \log{n}}$. They show an algorithm which matches this approximation factor. 
\arxivalt{\begin{observation}\label{card}
For the SLB problem, Algorithm EA and MUB are guaranteed to be no worse than factors of $O(\frac{\sqrt{n} \log n}{1 + (\sqrt{n} \log n - 1)(1 - \curv)})$ and $\frac{k}{1 + (k-1)(1 - \curv)}$ respectively.
\end{observation}
The guarantee for MUB follows directly from Corollary~\ref{SAAcorr}, by observing that $|X^*| = k$.
We also show a similar asymptotic hardness result, which is quite
close to the bounds in observation~\ref{card}.}{Corollaries~\ref{SAAcorr} and \ref{EAcorr} immediately imply \emph{curvature-dependent} approximation bounds of $\frac{k}{1 + (k-1)(1 - \curv)}$ and $O(\frac{\sqrt{n} \log n}{1 + (\sqrt{n} \log n - 1)(1 - \curv)})$.}
These bounds are improvements over the results of \cite{svitkina2008submodular} whenever $\kappa_f < 1$. Here, MUB is preferable to EA whenever $k$ is small. \arxivalt{The following theorem shows that the bound for EA is tight up to poly-log factors.\looseness-1
\begin{theorem}
For $\curv < 1$ and any $\epsilon > 0$, there exists submodular functions with curvature $\curv$ such that no poly-time algorithm achieves an approx. factor of $\frac{n^{1/2 - \epsilon}}{1 + (n^{1/2 - \epsilon} - 1)(1 - \curv)}$ for the SLB problem.\looseness-1
% he cardinality upper bound submodular minimization problem.
\end{theorem}}{Moreover, the bound of EA is tight up to poly-log factors, in that no polynomial time algorithm can achieve a general approximation factor better than $\frac{n^{1/2 - \epsilon}}{1 + (n^{1/2 - \epsilon} - 1)(1 - \curv)}$ for any $\epsilon > 0$. %~\cite{nips2013extendedvcurv}.

}
\arxiv{\begin{proof}
The proof of this theorem is analogous to that of theorem~\ref{learnhardness}. Define two monotone submodular functions $h_{\kappa}(X) = \curv\min\{|X|, \alpha\} + (1 - \curv)|X|$ and $f^R_{\kappa}(X) = \curv\min\{\beta + |X| \cap \bar{R}|, |X|, \alpha\} + (1 - \curv) |X|$, where $R \subseteq V$ is %chosen uniformly at random and has
% being 
a random set of 
cardinality $\alpha$. Let $\alpha$ and $\beta$ be an integer such that $\alpha = x\sqrt{n}/5$ and $\beta = x^2/5$ for an $x^2 = \omega(\log n)$. Also we assume that $k = \alpha$ in this case. Both $h_{\kappa}$ and $f^R_{\kappa}$ have curvature $\curv$. Given an arbitrary $\epsilon>0$, set
%We assume that 
$x^2 = n^{2\epsilon} = \omega(\log n)$. % for an arbitrary $\epsilon > 0$.
Then the ratio between $f^R_{\kappa}$ and $g^{\curv}$ is $\frac{n^{1/2-\epsilon}}{1 + (n^{1/2-\epsilon}-1)(1 - \curv)}$. Clearly then if any algorithm can achieve better than this bound, it can distinguish between $f_R$ and $g$ which is a contradiction.
\end{proof}
}
In the following problems, our ground set $V$ consists of the set of edges in a graph $\Gs = (\Vs,\Es)$ with two distinct nodes $s, t \in V$ and $n = |\Vs|$, $m = |\Es|$. The submodular function is $f: 2^{\Es} \rightarrow \mathbb{R}$.

\textbf{Shortest submodular s-t path (SSP).}
Here, we aim to find an s-t path $X$ of minimum (submodular) length $f(X)$.
%The problem here is to minimize a submodular function $f(X)$ such that $X$ is a path in a graph and $X \subseteq \Es$. 
\citet{goel2009approximability} show a $O(n^{2/3})$-approximation with matching curvature-independent lower bound $\Omega(n^{2/3})$. By Corollary~\ref{SAAcorr}, the curvature-dependent worst-case bound for MUB is $\frac{n}{1 + (n-1)(1 - \curv)}$ since any minimal s-t path has at most $n$ edges. Similarly, the factor for EA is $O(\frac{\sqrt{m}\log m}{1 + (\sqrt{m}\log m - 1)(1 - \curv)})$. The bound of EA will be tighter for sparse graphs while MUB provides better results for dense ones. 
% a hardness of $\Omega(n^{2/3})$ and a matching upper bound of $O(n^{2/3})$ for this problem. It trivially follows from Corollary~\ref{SAAcorr} that MUB achieves a worst-case bound of $\frac{n}{1 + (n-1)(1 - \curv)}$ since any minimal s-t path will have at most $n$ edges and $|X| \leq n$. Similarly, EA achieves a factor of $O(\frac{\sqrt{m}\log m}{1 + (\sqrt{m}\log m - 1)(1 - \curv)})$. The bound of EA will be tighter for sparse graphs while MUB provides better results for dense ones. 
\arxivalt{We can also show the following curvature-dependent lower bound:
\begin{theorem}
Given a submodular function with a curvature $\curv > 0$ and any $\epsilon > 0$, no polynomial-time algorithm achieves an approximation factor better than $\frac{n^{2/3 - \epsilon}}{1 + (n^{2/3 - \epsilon} - 1)(1 - \curv)}$ for the SSP problem.\looseness-1
\end{theorem}}{Our curvature-dependent lower bound for SSP is 
%We can also show that no poly-time algorithm can acheive a factor better than 
$\frac{n^{2/3 - \epsilon}}{1 + (n^{2/3 - \epsilon} - 1)(1 - \curv)}$, for any $\epsilon > 0$, which reduces to the result in \citep{goel2009approximability} for $\curv=1$.}
% ~\cite{nips2013extendedvcurv}.}
\arxiv{
\begin{proof}
The proof of this follows in very similar lines to the earlier lower bounds using our construction and the matroid constructions in~\cite{goel2009approximability}. The main idea is to use their multilevel graph, but define adjusted versions of their cost functions. In particular, define $h(X) = \curv\min\{|X|, \alpha\} + (1 - \curv)|X|$ and $f_R(X) = \curv\min\{\beta + |X| \cap \bar{R}|, |X|, \alpha\} + (1 - \curv) |X|$. In this context $R$ is a randomly chosen s-t path of length $n^{2/3}$ and $\alpha = n^{2/3}$. Similarly the value of $\beta = n^{\epsilon}$. The Chernoff bounds then show that the two functions above are indistinguishable (with high probability) and hence the ratio of the two functions $h$ and $f_R$ then provides the hardness result.
\end{proof}}

\textbf{Minimum submodular s-t cut (SSC): }
This problem, also known as the cooperative cut problem~\cite{jegelka2011-inference-gen-graph-cuts,jegelka2011-nonsubmod-vision}, asks to minimize a monotone submodular function $f$ such that the solution $X \subseteq \Es$ is a set of edges whose removal disconnects $s$ from $t$ in $\mathcal{G}$. \arxivalt{Using curvature refines the lower bound in \cite{jegelka2011-inference-gen-graph-cuts}:
%We can also prove the following curvature-dependent lower bound\notarxiv{ \cite{nips2013extendedvcurv}} that refines the lower bound in \cite{jegelka2011-inference-gen-graph-cuts}:
\begin{theorem}\label{thm:lowerbound_cut}
  % Given a submodular function with a curvature $\curv > 0$ and any $\epsilon > 0$, 
No polynomial-time algorithm can achieve an approximation factor better than 
$\frac{n^{1/2 - \epsilon}}{1 + (n^{1/2 - \epsilon} - 1)(1 - \curv)}$, for any $\epsilon > 0$, 
for the SSC problem with a submodular function of curvature $\curv$.\looseness-1
\end{theorem}}{Using curvature refines the We can also show a lower bound of \cite{jegelka2011-inference-gen-graph-cuts} to $\frac{n^{1/2 - \epsilon}}{1 + (n^{1/2 - \epsilon} - 1)(1 - \curv)}$, for any $\epsilon > 0$.}
\arxiv{
\begin{proof}
This proof follows along the lines of the results shown above. It uses the construction from~\cite{jegelka2011-inference-gen-graph-cuts}.
\end{proof}
}
Corollary~\ref{SAAcorr} implies an approximation factor of $O(\frac{\sqrt{m} \log m}{(\sqrt{m} \log m - 1)(1 - \curv) + 1})$ for EA and a factor of $\frac{m}{1 + (m-1)(1 - \curv)}$ for MUB, where $m=|\Es|$ is the number of edges in the graph. \arxivalt{By Theorem~\ref{thm:lowerbound_cut},}{Hence} the factor for EA is tight for sparse graphs. 
%
%Observe that algorithm EA, which provides a factor of $O(\frac{\sqrt{m} \log m}{(\sqrt{m} \log m - 1)(1 - \curv) + 1})$, meets the hardness bound for sparse graphs (i.e $m = \theta(n)$). Unfortunately the bound of algorithm MUB, which is $\frac{m}{1 + (m-1)(1 - \curv)}$ is quite loose. However another useful approximation, based on the polymatroidal network flow construction of~\cite{jegelka2011-inference-gen-graph-cuts} (which we call Algorithm PNA) provides a $n$ dependent bound. \notarxiv{For details on the polymatroidal construction, refer to~\cite{nips2013extendedvcurv}.} Algorithm PNA improves on the bound of EA when the graph is dense.\looseness-1
Specifically for cut problems, there is yet another useful surrogate function that is exact on local neighborhoods. \citet{jegelka2011-inference-gen-graph-cuts} demonstrate how this approximation may be optimized via a generalized maximum flow algorithm that maximizes a \emph{polymatroidal network flow} \cite{lawler82}. This algorithm still applies to the combination $\hat{f} = \curv \hat{f}^{\kappa} + (1-\curv)f^m$, where we only approximate $f^\kappa$. We refer to this approximation as Polymatroidal Network Approximation (PNA).
%
%it is possible to use an alternative surrogate function $\hat{f}$, where we approximate $f^{\kappa}$ by partitioning the ground set and computing a locally exact approximation. This corresponds to the polymatroidal flow construction in \cite{jegelka2011-inference-gen-graph-cuts}. The resulting $\hat{f}$ can still be optimized via a polymatroidal network flow. We name this approach Polymatroidal Network Approximation (PNA). Its approximation factor follows from combining results from \cite{jegelka2011-inference-gen-graph-cuts}, Theorem~\ref{thm:f_and_g} and Lemma~\ref{approxguarantee}:\looseness-1
\begin{corollary}
Algorithm PNA achieves a worst-case approximation factor of $\frac{n}{2 + (n-2)(1 - \curv)}$ for the cooperative cut problem.
% of minimizing $f$ under the constraints of being a s-t cut in a graph.
\end{corollary}
For dense graphs, this factor is theoretically tighter than that of the EA approximation.
\arxiv{\begin{proof}
We use the polymatroidal network flow construction from~\cite{jegelka2011-inference-gen-graph-cuts}, where the approximation $\hat{f}$ is defined via a partition of the ground set, and is separable over groups of edges. This approximation can be solved efficiently via generalized flows in polynomial time~\cite{jegelka2010online, jegelka2011-inference-gen-graph-cuts}. Moreover adding a modular term (for the modulation) does not increase the complexity of the problem.

This approximation satisfies $f^{\kappa}(X) \leq \hat{f^{\kappa}}(X) \leq \frac{n}{2} f^{\kappa}(X)$ for all cuts $X \in \mathcal{C}$.%\STR{Strictly speaking, this only holds for cuts, otherwise it is factor $n$.} 
We then convert this expression in the form of Theorem~\ref{thm:f_and_g} as $\frac{2\hat{f^{\kappa}}(X)}{n} \leq f^{\kappa}(X) \leq \hat{f^{\kappa}}(X)$. Then define $\hat{f}(X) \triangleq \curv \frac{2\hat{f^{\kappa}(X)}}{n} + (1 - \curv) \sum_{j \in X} f(j)$, and using theorem~\ref{thm:f_and_g}, it implies that:
\begin{align}
\hat{f}(X) \leq f(X) \leq \frac{n}{2 + (n - 2)(1-  \curv)} f(X)
\end{align}
Then let $\hat{X}$ be the minimizer of $\hat{f}(X)$ over $\mathcal C$ (using the generalized flows~\cite{jegelka2011-inference-gen-graph-cuts}). It then follows that (let $\alpha = \frac{n}{2 + (n - 2)(1-  \curv)}$): $f(\hat{X}) \leq \alpha \hat{f}(\hat{X}) \leq \alpha \hat{f}(X^*) \leq f(X^*)$ where $X^*$ is the optimal solution of $f$ over $\mathcal C$.
\end{proof}}

\textbf{Minimum submodular spanning tree (SST). }
Here, $\mathcal{C}$ is the family of all spanning trees in a given graph $\Gs$.
% This problem, introduced by~\cite{goel2009approximability} asks for minimizing a submodular function defined over the edge set, such that the set of edges forms a spanning tree. 
Such constraints occur for example in power assignment problems~\cite{wan02networks}. \citet{goel2009approximability} show a curvature-independent optimal approximation factor of $O(n)$ for this problem.
%We show that we
%can obtain better bounds for submodular functions with $\curv <
%0$. 
%Again, we can use the $O(n)$ algorithm
%of~\cite{goel2009approximability}, and provide a improved 
%$\curv$ dependent bound. 
\arxiv{\begin{observation}
For the minimum submodular spanning tree problem, algorithm MUB achieves an approximation guarantee, which is no worse than $\frac{n - r}{1 + (n - r -1)(1 - \curv)}$, where $r$ is the number of connected components of $\Gs$. 
\end{observation}}
\arxiv{\begin{proof}
    This result follows directly from Corollary~\ref{SAAcorr} and the fact that $|X^*| = n-r$. 
 \end{proof} }
\notarxiv{Corollary~\ref{SAAcorr} refines this bound to $\frac{n}{1 + (n -1)(1 - \curv)}$ when using MUB; Corollary~\ref{EAcorr} implies a slightly worse bound for EA.} 
\arxiv{In this case, Algorithm EA in fact provides slightly worse guarantees. Moreover the bound for MUB is optimal:
% It is easy to see that algorithm MUB provides a bound of $\frac{n}{1 + (n -1)(1 - \curv)}$ for this problem. Algorithm EA however acheives slightly worse approximation factors. 
% The following theorem shows that the bound from MUB is tight.
\begin{theorem}
For the class of submodular functions with curvature $\curv < 1$, no poly-time algorithm can achieve an approximation factor of
$\frac{n^{1-3\epsilon}}{1 + ( n^{1-3\epsilon} - 1)(1-  \curv) + \delta \curv}$ for the SST problem for any $\epsilon, \delta > 0$.
\end{theorem}}\notarxiv{We also show that the bound of MUB is tight: no polynomial-time algorithm can guarantee a factor better than $\frac{n^{1-\epsilon}}{1 + ( n^{1-\epsilon} - 1)(1-  \curv) + \delta \curv}$, for any $\epsilon, \delta > 0$.}
\arxiv{\begin{proof}
In this case, we use the construction of~\cite{goel2009approximability}, and define $f_{\kappa}^R(X) = \curv \min\{ |X \cap \bar{R}| + min\{|X \cap R|, \beta\}, \alpha\} + (1 - \curv) |X|$, and $g^{\curv}(X) = \curv min\{|X|, \alpha\} + (1- \curv) |X|$, where $\alpha = n^{ 1 + \epsilon}$, $\beta = n^{3\epsilon}(1 + \delta)$ and $|R| = \alpha$. For the formal graph construction, see~\cite{goel2009approximability}. Then with high probability $R$ is connected in the graph~\cite{goel2009approximability}. Since $f_R$ and $g$ are indistinguishable with high probability, so are $f_{\kappa}^R$ and $g^{\curv}$. Then notice that the minimum value of $f_{\kappa}^R$ and $g^{\curv}$ are $\curv\beta + (1-  \curv) n$ and $n$ respectively, and it is clear that the ratio between them is better than $\frac{n^{1-3\epsilon}}{1 + (n^{1-3\epsilon} - 1)(1 - \curv) + \delta \curv}$. Hence if any algorithm performs better than this, it will be able to distinguish $f_R$ and $g$ with high probability, which is a contradiction.
\end{proof}
Ana analogous analysis applies to 
% The above analysis also works with other similar forms of 
combinatorial constraints like Steiner trees~\cite{goel2009approximability}. 
}

\textbf{Minimum submodular perfect matching (SPM): }
Here, we aim to find a perfect matching in a graph that minimizes a monotone submodular function.
%Similar to the spanning tree, this problem asks for minimizing a submodular function over the set of perfect matchings in a graph. 
Corollary~\ref{SAAcorr} implies that an MUB approximation will achieve an approximation factor of at most $\frac{n}{2 + (n-2)(1 - \curv)}$.
%\arxiv{\begin{observation}
%Under the constraints of perfect matchings, algorithm MUB is guaranteed to be no worse than a factor of $\frac{n}{2 + (n-2)(1 - \curv)}$. 
%\end{observation}}
%\arxiv{\begin{proof}
%This theorem follows again directly from corollary~\ref{SAAcorr}, and the fact that $|X^*| = \frac{n}{2}$. 
%\end{proof}}
\arxiv{This bound is also tight:\looseness-1
\begin{theorem}
Given a submodular function with a curvature $\curv > 0$ and any $\epsilon > 0$, no polynomial-time algorithm achieves an approximation factor better than 
% Given a submodular function, with curvature $\curv$ and any $\epsilon, \delta > 0$, there cannot exist a polynomial time approximation algorithm, which achieves an approximation better than 
$\frac{n^{1 -3\epsilon}}{2 + (n^{1 -3\epsilon} - 2)(1 - \curv) + 2\delta \curv}$ for the SPM problem.\looseness-1
\end{theorem}}\notarxiv{Similar to the spanning tree case, the bound of MUB is also tight~\cite{nips2013extendedvcurv}.}
\arxiv{\begin{proof}
We use the same submodular functions as the spanning tree case, and it can be shown~\cite{goel2009approximability} that with high probability the set $R$ contains a perfect matching and the two functions are indistinguishable. Taking the ratio of $g^{\curv}$ and $f_R^{\curv}$, provides the above result.
\end{proof}}

\begin{figure}[t]
  \centering
\hspace{-10pt}
\includegraphics[width=0.25\textwidth]{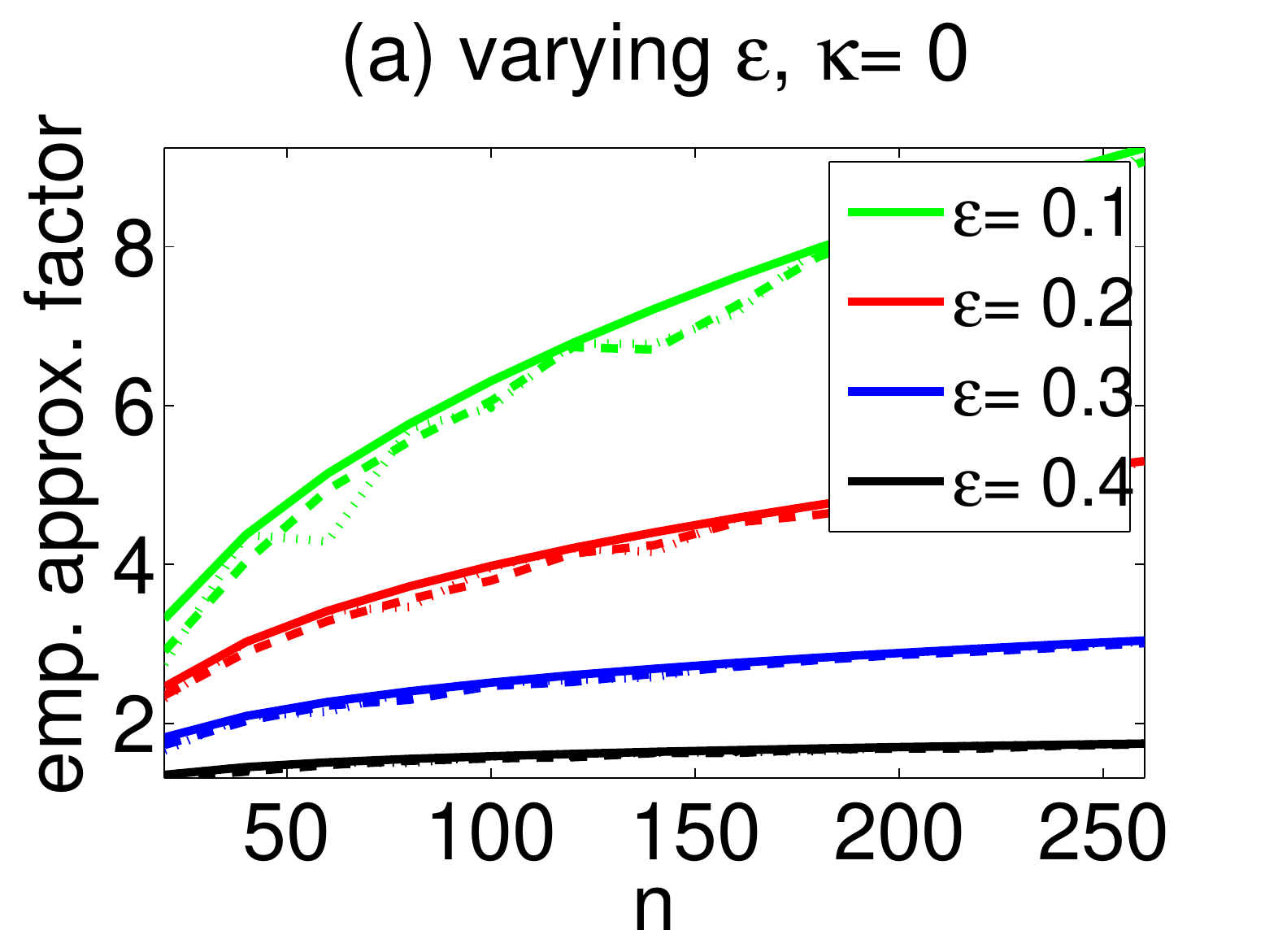}\hspace{-10pt} 
  ~ 
\includegraphics[width=0.25\textwidth]{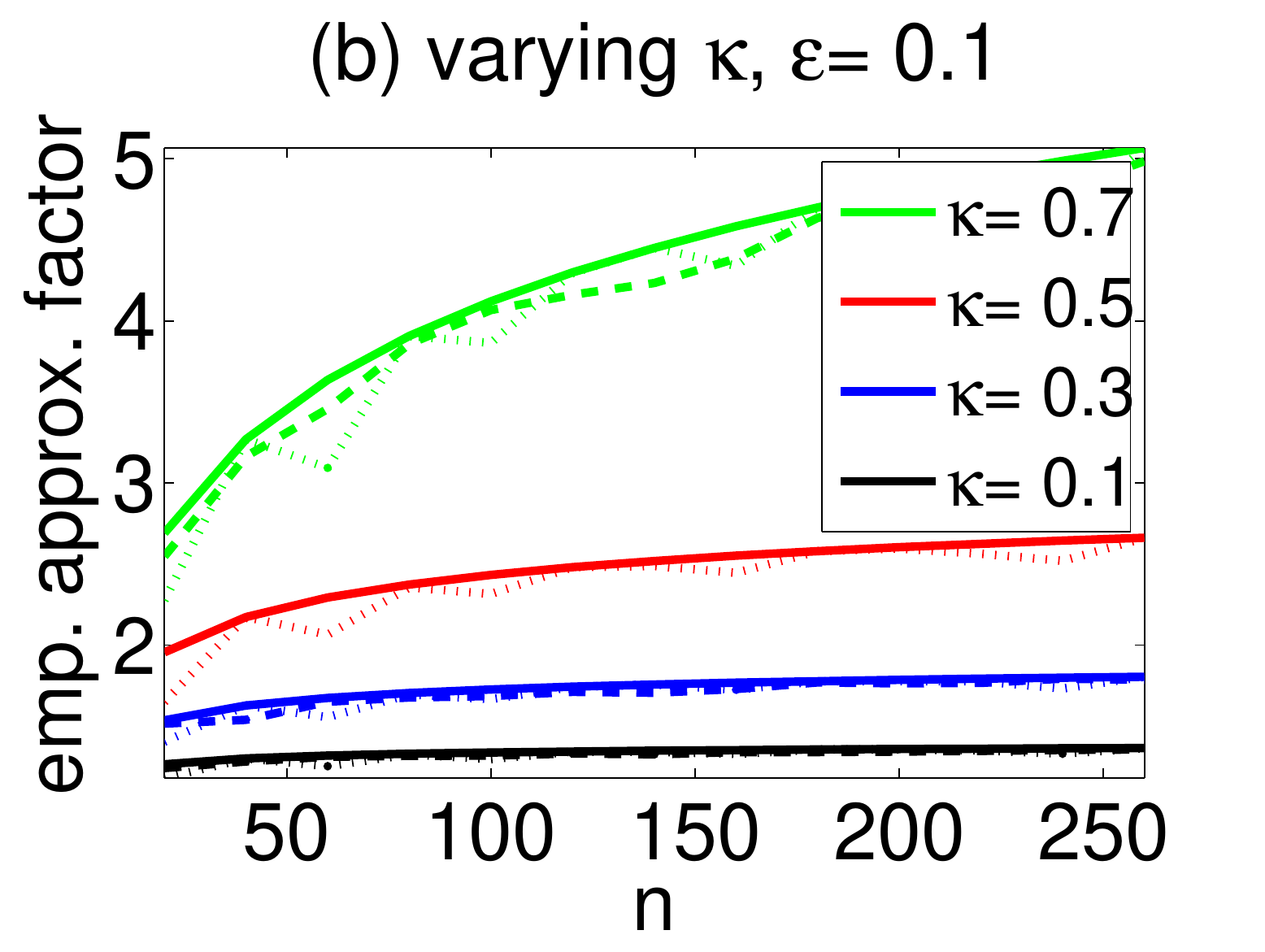} \hspace{-10pt} 
~
\includegraphics[width=0.25\textwidth]{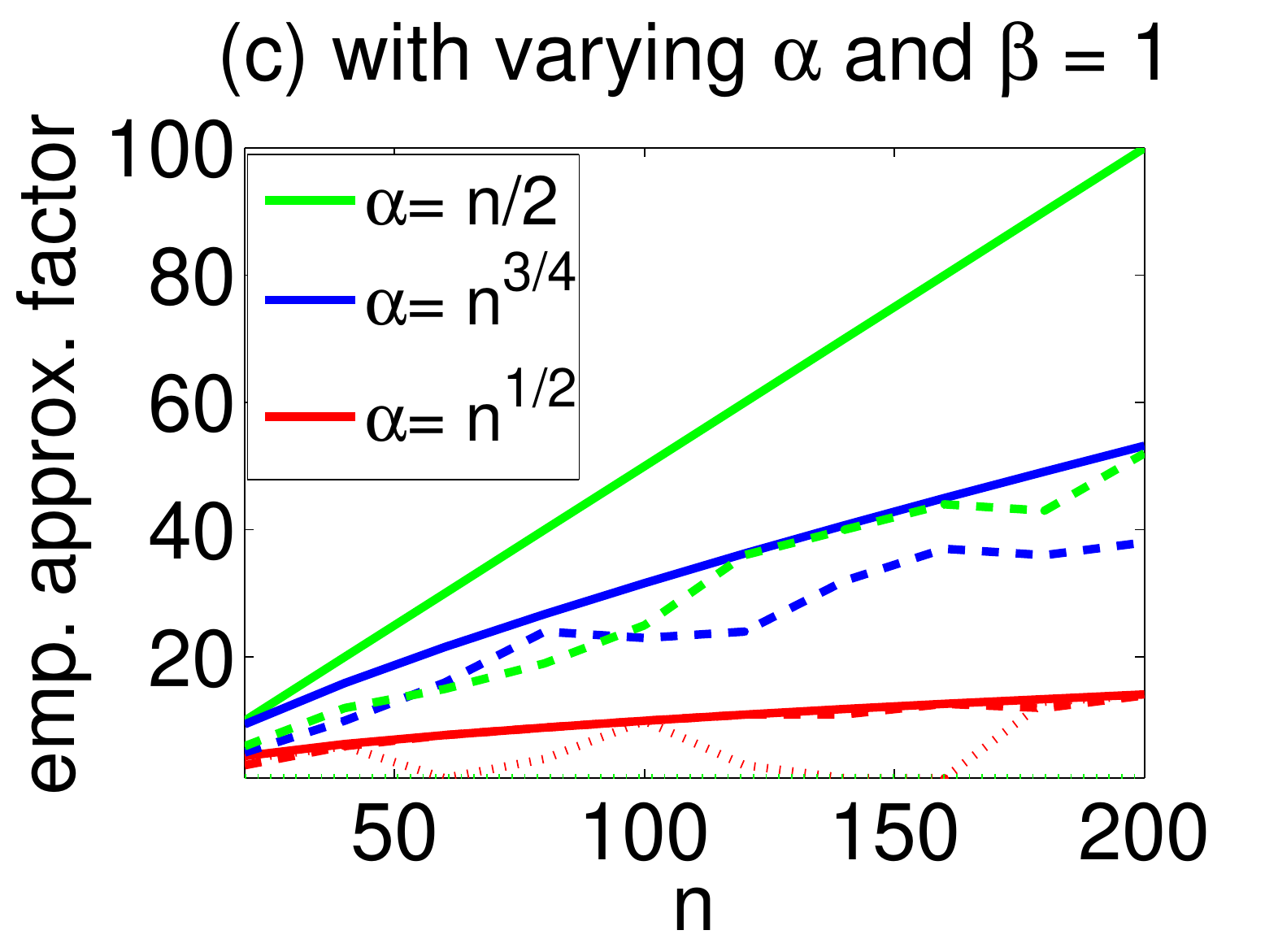}\hspace{-10pt} 
  ~ 
\includegraphics[width=0.25\textwidth]{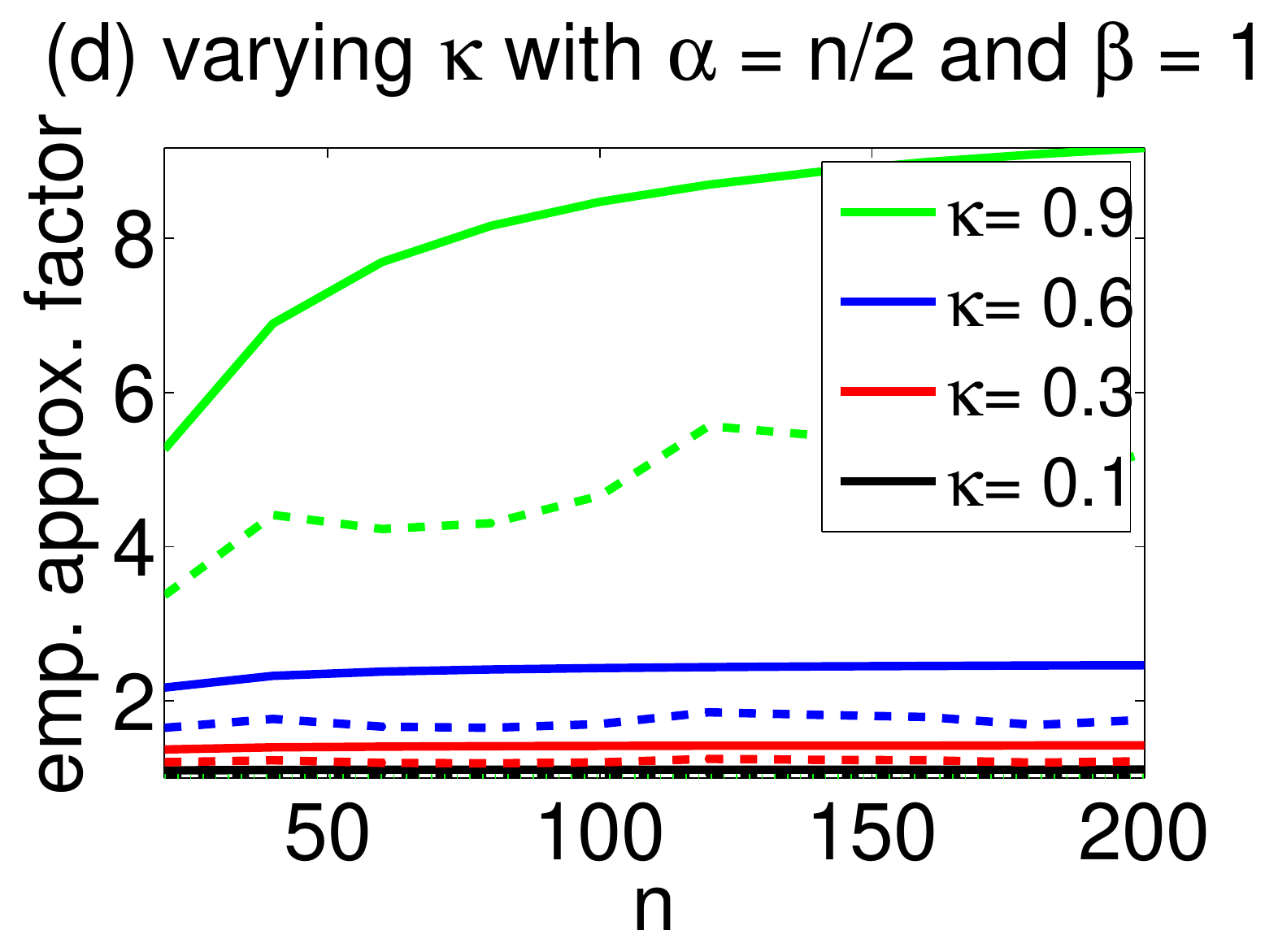} \hspace{-10pt} 
  \caption{Minimization of $g_\kappa$\arxiv{ (Equation~\ref{eq:5})} for cardinality lower bound constraints. (a) fixed $\kappa=0$, $\alpha = n^{1/2 + \epsilon}, \beta = n^{2\epsilon}$ for varying $\epsilon$; (b) fixed $\epsilon = 0.1$, but varying $\kappa$; (c) different choices of $\alpha$ for $\beta = 1$; (d) varying $\kappa$ with $\alpha = n/2, \beta = 1$. Dashed lines: MUB, dotted lines: EA, solid lines: theoretical bound. The results of EA are not visible in some instances since it obtains a factor of $1$.}
  \label{fig:consmin}
\end{figure}
\arxiv{\paragraph{\textbf{Minimum submodular edge cover:}}
The minimum submodular edge cover involves finding an edge cover (subset of edges covering all vertices), with minimum submodular cost. This problem has been investigated in~\cite{iwata2009submodular}, and they show that this problem is $O(n)$ hard. Algorithm MUB provides an approximation guarantee which is no worse than $\frac{2n}{2 + (n-2)(1- \curv)}$. 
We can show a almost matching hardness lower bound for this problem.
\begin{theorem}
Given a submodular function, with curvature coefficient $\curv$ and any $\epsilon, \delta > 0$, there cannot exist a polynomial-time approximation algorithm, which achieves an approximation better than $\frac{n^{1 -3\epsilon}}{2 + ( n^{1 -3\epsilon} - 2)(1 - \curv) + 2\delta \curv}$ for the minimum submodular edge cover problem.
\end{theorem}
\arxiv{\begin{proof}
We can use the construction of~\cite{iwata2009submodular} to show this. However a simple observation shows that a perfect matching is also an edge cover, and hence the hardness of edge cover has to be at least as much as the hardness of perfect matchings.
\end{proof}}}

\subsection{Experiments} We end this section by empirically demonstrating the performance of MUB and EA and their precise dependence on curvature. We focus on 
cardinality lower bound constraints, $\mathcal{C} = \{X \subseteq V: |X| \geq \alpha\}$ 
\JTR{above this was defined in terms of $k$ not $\alpha$. Fix.}\RTJ{Yes, but here $k = \alpha$. This is the notation used in the hardness function so probably we should keep it. I am with changing it also though.}
and
the ``worst-case'' class of functions that has been used throughout this paper to prove lower bounds, % We define the submodular function
\arxivalt{\begin{align}\label{eq:hardfct}
f^R(X) = \min\{|X \cap \bar{R}| + \beta, |X|, \alpha\},
\end{align}}{$f^R(X) = \min\{|X \cap \bar{R}| + \beta, |X|, \alpha\}$ }
where $\bar{R} = V \backslash R$ and $R \subseteq V$ is random set such that $|R| = \alpha$. We adjust $\alpha = n^{1/2 + \epsilon}$ and $\beta = n^{2 \epsilon}$ by a parameter $\epsilon$. The smaller $\epsilon$ is, the harder the problem.
\arxivalt{The function (\ref{eq:hardfct})}{This function} has curvature $\curv = 1$. To obtain a function with specific curvature $\kappa$, we define
\arxivalt{\begin{align}
  \label{eq:5}
  f^R_\kappa(X) = \kappa f(X) + (1-\kappa)|X|.
\end{align}}{$f^R_\kappa(X) = \kappa f(X) + (1-\kappa)|X|$ as in Equation~\eqref{eq:hidingfuncs}.}

In all our experiments, we take the average over $20$ random draws of $R$. We first set $\kappa=1$
%\STR{Is it $\kappa=1$ or $\kappa=0$ in (a)?} 
and vary $\epsilon$. Figure~\ref{fig:consmin}(a) shows the empirical approximation factors obtained using EA and MUB, and the theoretical bound. The empirical factors follow the theoretical results very closely. Empirically, we also see that the problem becomes harder as $\epsilon$ decreases.
Next we fix $\epsilon = 0.1$ and vary the curvature $\kappa$ in $f^R_\kappa$. Figure~\ref{fig:consmin}(b) illustrates that the theoretical and empirical approximation factors improve significantly as $\kappa$ decreases. Hence, much better approximations than the previous theoretical lower bounds are possible if $\kappa$ is not too large. This observation can be very important in practice. Here, too, the empirical upper bounds follow the theoretical bounds very closely.

Figures~\ref{fig:consmin}(c) and (d) show results for larger $\alpha$ and $\beta=1$. In Figure~\ref{fig:consmin}(c), as $\alpha$ increases, the empirical factors improve. In particular, as predicted by the theoretical bounds, EA outperforms MUB for large $\alpha$ and, for $\alpha \geq n^{2/3}$, EA finds the optimal solution. In addition, Figures~\ref{fig:consmin}(b) and (d) illustrate the theoretical and empirical effect of curvature: as $n$ grows, the bounds saturate and approximate a constant $1/(1-\kappa)$ -- they do not grow polynomially in $n$. 
% Figure~\ref{fig:consmin}(d) shows the effect of varying curvature, like Figure~\ref{fig:consmin}(b). This case is however much easier for EA than the modular scheme MUB. This is in contrast to figure~\ref{fig:consmin}(b) where both EA and MUB perform equally poorly, with rates matching the hardness. In many cases in figure~\ref{fig:consmin}(c, d) the dotted line of EA is flat at $1$. 
Overall, we see that the empirical results quite closely follow our theoretical results, and that, as the theory suggests, curvature significantly affects the approximation factors.

\section{Conclusion and Discussion}
In this paper, we study the effect of curvature on the problems of approximating, learning and minimizing submodular functions under constraints.
We prove tightened, curvature-dependent upper bounds 
with almost matching lower bounds. These results complement known results for submodular maximization~\cite{conforti1984submodular,vondrak2010submodularity}. \arxiv{Moreover, in~\cite{nipssubcons2013}, we also consider the role of curvature in submodular optimization problems over a class of \emph{submodular} constraints.}
Given that the functional form and effect of the submodularity ratio proposed in \cite{das2011submodular} is similar to that of curvature, an interesting extension is the question of whether there is a single unifying quantity for both of these terms. Another open question is whether a quantity similar to curvature can be defined for subadditive functions, thus refining the results in \arxivalt{\cite{balcan2011learning,badanidiyuru2012sketching}}{\cite{balcan2011learning}} for learning subadditive functions. Finally it also seems that the techniques in this paper could be used to provide improved curvature-dependent regret bounds for constrained online submodular minimization~\cite{jegelka2010online}.

{\bf Acknowledgments:} Special thanks to Kai Wei for pointing out that
Corollary~\ref{concvmodres} holds and for other discussions, to Bethany
Herwaldt for reviewing an early draft of this manuscript, and to the
anonymous reviewers.  This material is based upon work supported by
the National Science Foundation under Grant No. (IIS-1162606), a
Google and a Microsoft award, and by the Intel Science and Technology
Center for Pervasive Computing.  Stefanie Jegelka's work is supported
by the Office of Naval Research under contract/grant number
N00014-11-1-0688, and gifts from Amazon Web Services, Google, SAP,
Blue Goji, Cisco, Clearstory Data, Cloudera, Ericsson, Facebook,
General Electric, Hortonworks, Intel, Microsoft, NetApp, Oracle,
Samsung, Splunk, VMware and Yahoo!.
%\newpage
\notarxiv{\small}
\bibliographystyle{abbrvnat}
\bibliography{../Combined_Bib/submod}
\end{document}